\documentclass[a4paper,cleveref,thm-restate]{lipics-v2021}

\usepackage{style}
\usepackage{shortcuts}

\nolinenumbers

\hideLIPIcs

\title{Robustifying Sparse Matrix Multiplication}
\author{Karl Bringmann}{ETH Zurich}{karl.bringmann@inf.ethz.ch}{}{}
\author{Nick Fischer}{Max Planck Institute for Informatics}{nfischer@mpi-inf.mpg.de}{}{}
\author{Vasileios Nakos}{National and Kapodistrian University of Athens}{vasilisnak@di.uoa.gr}{}{}
\authorrunning{K. Bringmann, N. Fischer and V. Nakos}

\begin{document}

\maketitle

\begin{abstract}
\noindent
In the seminal sparse matrix multiplication problem the goal is to compute the product of two~\makebox{$n \times n$} matrices when the matrices are sparse, i.e., when the number of nonzeros in the input matrices~$\IN$ and/or the number of nonzeros in the output matrix $\OUT$ are much smaller than~$n^2$. In this paper, we explore the generalized problem of (approximately) computing the~$k$ \emph{largest} output entries, with an approximation error dependent solely on the smaller entries---from the viewpoint of sparse recovery, this can be seen as a \emph{robust} variant of sparse matrix multiplication. Despite the substantial research dedicated to sparse matrix multiplication, almost no existing algorithms are robust in this sense. The one exception is Pagh's algorithm in time~\smash{$\widetilde\Order(\IN + nk)$} [ITCS~'12], and it remained open whether other algorithms can be similarly made robust.

Our principal contribution is a \emph{black-box reduction} from robust sparse matrix multiplication to conventional sparse matrix multiplication with only polylogarithmic overhead. Specifically, we show that any sparse matrix multiplication algorithm with running time $T(n, \IN, \OUT)$ can be transformed into a robust algorithm running in time~\smash{$\widetilde\Order(T(n, \IN, k))$}. This reduction leverages an extensive toolkit from sparse recovery, and intriguingly, also involves solving a knapsack-type problem.

By plugging in the state-of-the-art algorithm for sparse matrix multiplication by Abboud, Bringmann, Fischer, and Künnemann [SODA'24], we achieve significantly improved bounds such as $\Order((\IN + k)^{1.346})$. Notably, in the regime where~\smash{$k \geq \IN^{1.762}$}, our reduction culminates in an \emph{almost-optimal} $k^{1+\order(1)}$-time algorithm.
\end{abstract}


\section{Introduction} \label{sec:intro}
Few problems have sparked as much theoretical and practical research in computer science as the quest for efficient \emph{matrix multiplication} algorithms. Motivated by countless applications, several communities have contributed to the development of Strassen-like algebraic algorithms (see~\cite{AlmanDWXXZ25} for the current record $\omega < 2.3714$ on the matrix multiplication exponent), lower bounds~\cite{Blaser99,Raz03,Shpilka03,Landsberg14}, and specialized practically fast algorithms~\cite{BallardDHLS12a,GaoJCHWLW23}. An especial focus lies on \emph{sparse} matrix multiplication, where the goal is to achieve faster algorithms when either the \emph{input-sparsity $\IN$} (the number of nonzero entries in $A$ and $B$) or the \emph{output-sparsity $\OUT$} (the number of nonzero entries in~$AB$) is much smaller than $n^2$. This problem has been studied extensively~\cite{Gustavson78,YusterZ05,AmossenP09,Lingas11,Pagh13,Kutzkov13,WilliamsY14,JacobS15,GuchtWWZ15,GasieniecLLPT17,Kunnemann18,Roche18a,DeepHK20,AbboudBFK24,BennettGGW24,BennettGGW25}, and many faster algorithms in terms of the three parameters~$n, \IN, \OUT$ have been proposed.

In this paper we consider the generalization of sparse matrix multiplication where the output matrix $A B$ is not necessarily sparse, but consists of $k \ll n^2$ \emph{large, significant} entries plus up to~$n^2$ small, insignificant entries which we regard as \emph{noise}---our goal is to recover the $k$ significant entries (with some small approximation error depending on the noise). This problem is arguably very natural and well-motivated for various applications, for instance in a setting where the output matrix is sparse up to small measurement errors, or whenever the output matrix represents some numerical scores out of which only $k$ scores are expected to be significant; see e.g.~\cite{CohenL99,ErikssonBiqueSSWAI11,VegaGR12}.

A concrete application in the context of information retrieval and pattern recognition is to search a database of documents (e.g., searching for a phrase on wikipedia). This task is commonly modelled as follows; see e.g.~\cite{ErikssonBiqueSSWAI11} for more details.\footnote{This application has first been heuristically stated as an application of approximate matrix multiplication~\cite{CohenL99,ErikssonBiqueSSWAI11,VegaGR12} as discussed in \cref{sec:intro:sec:related-work}; however, upon closer inspection it turns out that this task is modelled more accurately by the robust sparse matrix multiplication problem described here. See also the discussion in~\cite{Pagh13}.} We fix a set of~$t$~\emph{terms} (i.e., important phrases in a search query) and store a database of $n$ documents by a \emph{term-by-document matrix}~\makebox{$A \in \Real^{n \times t}$}. Here each column is associated to a term, each row is associated to a document, and each entry~$A_{i, j}$ stores the frequency with which the term appears in the document. A search query is represented by a length-$t$ vector~$b$ of term frequencies, so that the inner product~$A b$ associates to each document a similarity score (called the \emph{cosine similarity})---the higher the score, the better the respective document matches the query. Therefore, evaluating multiple database queries is to compute the \emph{large} entries in the matrix product $A B$ (where the columns of $B$ are the queries). Another equally important application is to compute all pairs of similar documents, which amounts to computing the large entries of~$A A^T$~\cite{CohenL99}.

\subparagraph*{Robustification via Sparse Recovery}
Viewed through the lens of sparse recovery and compressed sensing, the problem we study is a natural \emph{robustification} of the exact sparse matrix multiplication problem. Similar robustifications are common, e.g.\ for the seminal \emph{sparse recovery} problem where the goal is to recover a $k$-sparse vector $x$ from as few linear measurements as possible. In the \emph{robust} sparse recovery problem the vector $x$ is not necessarily $k$-sparse, but instead consists of~$k$ significant entries plus some noise. Formally, the goal is to find a vector $x'$ such that
\begin{alignat*}{2}
    \norm{x - x'}_2 &\leq (1 + \epsilon) \cdot \norm{x_{-k}}_2 &&\qquad\text{(``$\ell_2/\ell_2$-guarantee'')},
\intertext{or that}
    \norm{x - x'}_\infty &\leq \frac{1}{\sqrt{k}} \cdot \norm{x_{-k}}_2 &&\qquad\text{(``$\ell_\infty/\ell_2$-guarantee'')}.
\end{alignat*}
Here, $x_{-k}$ is the vector obtained from $x$ after zeroing out the largest $k$ coordinates. In this way, as desired, the approximation error only depends on the \emph{insignificant} entries of $x$. The latter $\ell_\infty/\ell_2$-guarantee is strictly stronger,\footnote{It is easy to verify that $\ell_\infty/\ell_2$ is at least as strong as $\ell_2/\ell_2$; see for instance the proof of~\cite[Theorem~2.1]{price20111+}. To see that the guarantee is strictly stronger, consider a vector $x$ with, say, $0.05k$ coordinates equal to~\smash{$\frac{2}{\sqrt{k}}$} and where the total $\ell_2$-mass of the other coordinates is $1$. An algorithm with the $\ell_2/\ell_2$-guarantee (for $\epsilon = 0.1$, say) could as well detect no coordinate (i.e., $x' = 0$ is a feasible solution), whereas any algorithm with $\ell_\infty/\ell_2$-guarantee \emph{must} detect all the $0.05k$ special entries.} and also conceptually more desirable as it prevents missing relevant entries or false outliers. Both the $\ell_\infty/\ell_2$ and the $\ell_2/\ell_2$ are very extensively studied in the field of sparse recovery, including Sparse Fourier transforms~\cite{HassaniehIKP12, Kapralov16, IndykKP14, NakosSW19, KapralovVZ19}, sketching and streaming algorithms~\cite{KaneNPW11, larsen2019heavy}, and reconstruction of signals from partial measurements~\cite{candes2006robust, candes2006stable, GilbertLPS10, price20111+, IndykPW11, NakosS19, FischerN25}, to name a few from a vast list. In the same spirit we formally define the \emph{robust} version of sparse matrix multiplication as follows:

\begin{definition}[Robust Sparse Matrix Multiplication] \label{def:robust}
Given matrices $A, B \in \Real^{n \times n}$ in sparse representation, and given~$k \geq 1$, compute a matrix $C \in \Real^{n \times n}$ such that
\begin{equation*}
    \norm{A B - C}_\infty \leq \frac{1}{\sqrt{k}} \cdot \norm{(A B)_{-k}}_F.
\end{equation*}
Here, $(A B)_{-k}$ denotes the matrix obtained from $A B$ by zeroing out the largest $k$ entries (in absolute value).
\end{definition}

Note that this is indeed a generalization of sparse matrix multiplication (as by setting $k \geq \OUT$, we force the error to be $\norm{(A B)_{-k}}_F = 0$, and thus $C = AB$). The study of this problem was initiated in 2013 by Pagh~\cite{Pagh13} who developed the first robust sparse matrix multiplication algorithm in time $\widetilde\Order(\IN + k n)$.\footnote{Pagh~\cite{Pagh13} referred to the problem as \emph{compressed} matrix multiplication. We instead prefer the name \emph{robust sparse} matrix multiplication to be more in line with the sparse recovery terminology.} In a nutshell, his algorithm relies on various tools from sparse recovery and interestingly encodes the problem as the multiplication of two \emph{polynomials} (implemented by the Fast Fourier Transform). Among the vast collection of known sparse matrix multiplication algorithms~\cite{Gustavson78,YusterZ05,AmossenP09,Lingas11,Pagh13,Kutzkov13,WilliamsY14,JacobS15,GuchtWWZ15,GasieniecLLPT17,Kunnemann18,Roche18a,DeepHK20,AbboudBFK24} Pagh's algorithm remained the only algorithm which is robust in the sense of \cref{def:robust}, and none of the other algorithms can be easily robustified.

This leaves open whether we could benefit from the ideas and improvements in the other sparse matrix multiplication algorithms to also obtain improved bounds for the robust problem. Specifically, many of these algorithms exploit Strassen-like algebraic fast matrix multiplication---can we obtain improved robust algorithms using fast matrix multiplication?

\subsection{Main Result: A Reduction} \label{sec:intro:sec:main-result}
Motivated by this question we investigated which algorithms can be robustified. We arrived, satisfyingly, at a very general solution: Our main contribution is \emph{black-box reduction} from \emph{robust} sparse matrix multiplication to \emph{exact} sparse matrix multiplication with only polylogarithmic overhead: 

\begin{restatable}[Robustifying Reduction]{theorem}{thmrobustifying} \label{thm:robustifying}
Suppose that sparse matrix multiplication of $n \times n$ matrices (with $\IN$ nonzero inputs and $\OUT$ nonzero outputs) is in time $T(n, \IN, \OUT)$. Then robust sparse matrix multiplication is in time $\Order(n \log^4 n + \log^3 n \cdot T(n, \IN, k)) = \widetilde\Order(T(n, \IN, k))$ (by a randomized algorithm that succeeds with high probability).
\end{restatable}

The reduction naturally builds on the broad toolkit provided by sparse recovery. Vaguely speaking, our main challenge is to estimate how the large and small entries are distributed among the columns of $A B$. The innovation in our solution lies in encoding this task as an optimization problem---specifically, an instance of Minimum-Cost Multiple-Choice Knapsack---and observing that this problem can be efficiently approximated. In \cref{sec:intro:sec:overview} we give a more detailed overview.

As a consequence of \cref{thm:robustifying} we obtain improved algorithms for robust sparse matrix multiplication in almost all regimes. To state these improvements precisely we introduce some notation: Let $\MM(x, y, z)$ denote the running time of algebraic dense rectangular matrix multiplication (i.e., the complexity of multiplying an $x \times y$ matrix by a $y \times z$ matrix). The state-of-the-art of sparse matrix multiplication is a recent algorithm due to Abboud, Bringmann, Fischer and Künnemann~\cite{AbboudBFK24}. Taking all three parameters $n, \IN, \OUT$ into account this algorithm runs in time
\begin{equation*}
    \widetilde\Order\parens*{\max_{\substack{x, z \leq n\\xz \leq \OUT}} \min_{1 \leq \Delta \leq n} \parens*{\Delta \IN + \MM\parens*{x, \min\set*{\frac{\IN}{\Delta}, n}, z}}}.
\end{equation*}
This complicated running time can often be conveniently bounded with respect to the matrix multiplication constants, e.g., by $\Order((\IN + \OUT)^{1.346})$~\cite{AbboudBFK24}. In light of \cref{thm:robustifying} we immediately achieve the same running time for robust sparse matrix multiplication with $k$ in place of $\OUT$:

\begin{corollary} \label{cor:in-plus-k}
Robust sparse matrix multiplication is in time $\widetilde\Order((\IN + k)^{1.346})$ (by a randomized algorithm that succeeds with high probability).
\end{corollary}

This running time improves upon Pagh's $\widetilde\Order(\IN + n k)$-time algorithm in many regimes. Specifically, in the natural regime $\IN \leq k$, we obtain time $\widetilde\Order(k^{1.346}) = \widetilde\Order(n^{0.692} k)$ since $k \le n^2$. To allow for a clearer comparison, we also bound our running time as follows:

\begin{corollary} \label{cor:faster-than-pagh}
Robust sparse matrix multiplication is in time $\widetilde\Order(\IN + \MM(n, n, \ceil{\frac{k}{n}}) \cdot n^\epsilon)$ for any~\makebox{$\epsilon > 0$} (by a randomized algorithm that succeeds with high probability).\footnote{Let us remark that the overhead $n^\epsilon$ in \cref{cor:faster-than-pagh,cor:almost-optimal} is a technical artifact from the algebraic matrix multiplication machinery. Specifically, while many papers for the sake of simplicity state that matrix multiplication is in time $\Order(n^\omega)$, the technically correct statement is that for every $\epsilon > 0$, there is a matrix multiplication algorithm in time $\Order(n^{\omega+\epsilon})$. Here, as we are dealing with matrix multiplication very directly, we prefer the formally correct statement.}
\end{corollary}

That is, we can basically replace the term $n k$ by the time complexity of multiplying an $n \times n$ by an~\makebox{$n \times \ceil{\frac{k}{n}}$} matrix, which can naively be upper bounded by $\Order(n k)$. Less naively, by using algebraic matrix multiplication we have\footnote{This time bound is clear if $k \leq n$. Otherwise, we split both matrices into submatrices of size at most~\smash{$\ceil{\frac{k}{n}} \times \ceil{\frac{k}{n}}$}. To compute the original matrix product it then suffices to compute $\Order((n^2 / k)^2)$ square matrix products of size~\smash{$\ceil{\frac{k}{n}} \times \ceil{\frac{k}{n}}$} each of which runs in time $\Order((k / n)^{\omega})$. The total time is $\Order((n^2 / k)^2 \cdot (k / n)^\omega) = \Order(nk \cdot (n / k)^{3-\omega})$. We remark that this running time can be improved by applying rectangular matrix multiplication for any $k = \Theta(n^\gamma)$ with $1 < \gamma \leq 2$.}
\begin{equation*}
    \MM(n, n, \ceil{\frac{k}{n}}) = \Order(n^2 + nk \cdot (n / k)^{3-\omega}) = \Order(n^2 + nk \cdot (n / k)^{0.6284}),
\end{equation*}
yielding a polynomial improvement over Pagh's running time $\widetilde\Order(\IN + n k)$ whenever $k \geq n^{1+\epsilon}$ and~\smash{$nk \ge \IN^{1+\epsilon}$} (i.e., Pagh's algorithm does not already run in almost-linear time).

\smallskip
The Abboud--Bringmann--Fischer--Künnemann algorithm allows for various other trade-offs with respect to the relation between $\IN$ and $k$, see~\cite{AbboudBFK24}. Here we will only mention one last interesting setting: When $k$ is sufficiently large, namely $k \geq \IN^{1.762}$, then we obtain an \emph{unconditionally almost-optimal} algorithm:

\begin{corollary} \label{cor:almost-optimal}
If $k \geq \IN^{1.762}$, then robust sparse matrix multiplication is in time $\Order(k^{1+\epsilon})$ for any~\makebox{$\epsilon > 0$} (by a randomized algorithm that succeeds with high probability). 
\end{corollary}

Indeed, observe that any algorithm requires time $\Omega(k)$ simply to write the output and therefore the time in \cref{cor:almost-optimal} is optimal up to the factor $k^\epsilon$. This corollary is ultimately based on the well-known fact that matrix multiplication of sufficiently \emph{rectangular} matrices is in almost-optimal time (i.e., $\MM(n, n, n^{0.31389}) = \Order(n^{2+\epsilon})$ for any $\epsilon > 0$, where the constant 0.31389, sometimes called the \emph{dual matrix multiplication exponent}, is the state of the art in a series of improvements~\cite{Coppersmith82,Coppersmith97,Gall14,GallU18,VassilevskaXXZ24}).

\subsection{Related Work: Approximate Matrix Multiplication} \label{sec:intro:sec:related-work}
Our research is also closely related to the \emph{approximate matrix multiplication} problem. Here the goal is similarly to approximate $A B$ by some matrix $C$ in the same sense of bounding some norm of the matrix $A B - C$ (such as the Frobenius or spectral norm). The conceptual difference is that~$A B$ is not assumed to be sparse or sparse up to noise, and therefore the goal is not to bound the error in terms of the insignificant entries in $A B$ (i.e., $\norm{(A B)_{-k}}_F$) but rather in terms of all entries in $A B$ (i.e., say,~$\norm{A B}_F$). The typical guarantee is that one can compute some matrix $C$ in time~\smash{$\widetilde\Order(n^2 c)$}, for some parameter $c$, such that
\begin{equation*}
    \norm{A B - C}_F \leq \frac{1}{\sqrt c} \norm{A}_F \norm{B}_F.
\end{equation*}
This can be achieved by two seminal approaches: \emph{sampling}-based approaches due to Cohen and Lewis~\cite{CohenL99} and Drineas, Kannan and Mahoney~\cite{DrineasK01,DrineasKM06}, and \emph{projection}-based approaches as introduced by Sarl{\'{o}}s~\cite{Sarlos06}. These two approaches form the basis of a plethora of later refinements~\cite{Wu18,YangM23} and extensions to other models such as streaming~\cite{ClarksonW09}. There was also some recent progress on incorporating fast matrix multiplication techniques in this setup~\cite{UffenheimerW25}.

All of these works are incomparable to our results: Approximate matrix multiplication is well-suited for dense and \emph{flat} matrices (e.g., when most entries have roughly the same size) or whenever an $\ell_2 / \ell_2$ guarantee is sufficient, whereas robust sparse matrix multiplication is preferable for \emph{spiky} matrices or whenever an $\ell_\infty / \ell_2$ guarantee is necessary. Having said that, one can of course simply apply robust matrix multiplication algorithms to the approximate setting (disregarding the strong robust guarantee). For instance, Pagh applied his algorithm to this setting to achieve a bound of
\begin{equation*}
    \norm{A B - C}_F \leq \sqrt{\frac{n}{c}} \cdot \norm{A B}_F
\end{equation*}
in time $\widetilde O(n^2 c)$. This statement was also very recently obtained in an arXiv preprint~\cite{UffenheimerW26} (though notably without the stronger robust guarantee). When applying our faster algorithm we obtain the strictly stronger accuracy bound of
\begin{equation*}
    \norm{A B - C}_F \leq \sqrt{\frac{n}{c^{\frac{1}{\omega - 2+\epsilon}}}} \cdot \norm{A B}_F \leq \sqrt{\frac{n}{c^{2.963}}} \cdot \norm{A B}_F
\end{equation*}
in the same time budget $\widetilde O(n^2 c)$.\footnote{To see this, bound the running time of our algorithm from \cref{cor:faster-than-pagh} by $\MM(n, n, k/n) \leq O(n^4 / k^2 \cdot \MM(k/n)) \leq O(n^{4-\omega+\epsilon} k^{\omega+\epsilon-2})$, for any arbitrarily small constant $\epsilon > 0$. Then, setting $k$ such that $c = (k / n)^{\omega+\epsilon-2}$ the running time becomes $O(n^2 c)$ and the guarantee becomes as claimed before.} To reiterate, both Pagh's and our algorithm actually obtain stronger guarantees in terms of $\norm{(A B)_{-k}}_F$ for some $k \geq \Omega(n/c)$.

\subsection{Technical Overview} \label{sec:intro:sec:overview}
In this section we will briefly outline the main ideas behind our results. For illustration, we consider the following query problem. Let $X \in \mathbb{R}^{n \times n}$ be an unknown matrix. We can access $X$ by querying the product $L X R$ for a matrix $L \in \{-1,0,1\}^{\ell \times n}$ with column sparsity 1 and a matrix $R \in \{-1,0,1\}^{n \times r}$ with row sparsity 1 of our choice; the cost of this query is $\ell \cdot r$. The task is to compute an approximation $X'$ of~$X$ with entry-wise error at most \smash{$\frac 1{\sqrt{k}} \norm{X_{-k}}_F$}, using \smash{$\widetilde{O}(1)$} queries of total cost $\widetilde{O}(k)$.\footnote{As we will later see, it is often sufficient to only consider one side of the sketch (say $L$) and leave the other side (say $R$) trivial. Moreover, it is of course equivalent to assume that the column-sparsity of $L$ and row-sparsity of $R$ is $\tilde O(1)$ instead of $1$.}

\subparagraph*{Connection to Robust Matrix Multiplication}
Let us first explain the connection of the query problem to robust matrix multiplication. Suppose we want to multiply two $\IN$-sparse $n \times n$-matrices $A$ and $B$. The overall goal is to \emph{compress} the matrices~\makebox{$A, B \in \Real^{n \times n}$} to \emph{denser rectangular} matrices $A' \in \Real^{x \times n}, B' \in \Real^{n \times z}$ in such a way that the product of the outer dimensions becomes small, i.e., $x z \leq \widetilde{O}(k)$. Since the number of non-zeros in the matrix product~$A' B'$ is trivially bounded by $x z \leq \widetilde{O}(k)$, we can apply in a black-box manner any sparse matrix multiplication algorithm to evaluate $A' B'$. The challenge is to design the compression in such a way that we can recover the large entries of~$A B$ from~$A' B'$.
(This general framework has been implicitly or explictly employed in various forms in previous work on \emph{exact} matrix multiplication~\cite{IwenS09,Pagh13,JacobS15,Stockel15,GuchtWWZ15,AbboudBFK24}, specifically, compared to~\cite{AbboudBFK24} our goal is to replace the ``densification'' step by a robust alternative.)

The above query problem models this compression:
We set the unknown matrix to $X := AB$. We can answer a query $L X R$ as follows. First compute $L A$ in time $O(\IN)$, using that $L$ has column sparsity 1. Symmetrically compute $B R$ in time $O(\IN)$, using that $R$ has row sparsity~1. The remaining product $(L A)(B R)$ has outer dimensions $\ell$ and $r$. Thus, if each query has cost $\ell \cdot r = \widetilde{O}(k)$, then each product $(L A)(B R)$ is compressed in the above sense. Therefore, any algorithm for the query problem yields a reduction from robust matrix multiplication to sparse matrix multiplication. 

\medskip
In the remainder we discuss how to solve the query problem. 
 
\subparagraph*{Warm-Up: The Perfectly Balanced Case}
To demonstrate that the sparse recovery toolkit is applicable, consider first the toy case in which all columns in $X$ contain exactly the same number of significant entries, i.e., $t := k / n$ many. Suppose further that the noise is evenly distributed among all columns, i.e.,
\begin{equation*}
    \norm{(X e_j)_{-t}}_2^2 \leq \frac{\norm{X_{-k}}_F^2}{n}
\end{equation*}
for each column vector $X e_j$ of $X$ (here $e_j$ denotes the $j$-th unit vector).

The known $\ell_\infty/\ell_2$-sparse recovery algorithms, e.g.~\cite[Theorem~1]{KaneNPW11} can be phrased as follows (see \cref{lem:heavy-hitters}): There is a matrix~$H$ of size~\smash{$\widetilde\Order(t) \times n$} such that, given $H x$ for some unknown vector~\makebox{$x \in \Real^n$}, we can efficiently compute a vector $x'$ so that~\smash{$\norm{x - x'}_\infty \leq \frac{1}{\sqrt t} \|x_{-t}\|_2$} (with good probability). Moreover, the matrix $H$ has at most \smash{$\widetilde{O}(1)$} nonzero entries per column, and can be constructed in time~\smash{$\widetilde\Order(n)$}. A natural idea is to compute $HX$. Since $H$ has column sparsity \smash{$\widetilde{O}(1)$}, we can write it as a sum \smash{$H = L_1+\ldots + L_{\widetilde{O}(1)}$} where each $L_i$ has column sparsity 1. Then $L = L_i$ and the identity matrix $R = I_n$ is a valid query, and by summing up the query results over all $i$ we obtain $HX$.
This choice indeed satisfies our requirements: It uses $\widetilde{O}(1)$ many queries and, since $H$ has $\widetilde{O}(t)$ rows and $I_n$ has $n$ columns, each query has cost $\widetilde{O}(t) \cdot n = \widetilde{O}(k)$. 
Moreover, from each column vector in~\makebox{$H X$} we can efficiently recover an approximation $x'_j$ of the corresponding column vector~$X e_j$ in $X$ with an error of
\begin{equation} \label{eq:linfty-overview}
    \norm{X e_j - x'_j}_\infty \leq \frac{1}{\sqrt{t}} \cdot \norm{(X e_j)_{-t}}_2 \leq \frac{1}{\sqrt{tn}} \cdot \norm{X_{-k}}_F = \frac{1}{\sqrt k} \cdot \norm{X_{-k}}_F.
\end{equation}
This is exactly as hoped. However, we critically relied on the unrealistic assumption that both the~$k$ largest entries and the noise is evenly distributed among all columns.

\subparagraph*{The General Unbalanced Case}
In the general case we cannot wish for this perfect distribution. Instead, we have to deal with some columns containing more significant entries and/or more noise than others. It could even be that one column has only one significant entry and the $\ell_2$ mass of the rest of the elements is much larger than that entry; this means that we cannot hope to recover that single element with $\ell_\infty/\ell_2$ sparse recovery algorithm with $t=1$ but have to assign a much larger budget. It turns out that the ideas from the toy case would generalize if we would know both the number of significant entries $k_j$ as well as the noise $\norm{(X e_j)_{-k_j}}_2$ for each column~$j$. In this case we could split the columns into $\polylog(n)$ groups where in each group all columns share the same statistics (up to constant factors); then the previous approach applies to each group. So perhaps a first instinct is to estimate $k_j$ and $\norm{(X e_j)_{-k_j}}_2$ for each column in a preprocessing step. Unfortunately, this appears to be quite challenging---in particular, but not exclusively, because $k_j$ depends on the entire matrix and not only on the $j$-th column.

Our solution is a more modest implementation of this approach, where our goal is to assign to each column a \emph{budget} $t_j$. On the one hand, the budget should satisfy that
\begin{equation} \label{eq:accuracy}
    \frac{1}{\sqrt{t_j}} \cdot \norm{(X e_j)_{-t_j}}_2 \leq \frac{1}{\sqrt k} \cdot \norm{X_{-k}}_F;
\end{equation}
note that this is exactly the accuracy we require in~\eqref{eq:linfty-overview}. However, this condition can easily be enforced by assigning the generous budgets $t_j = n$ to each column. To avoid this we require, on the other hand, that the total budget is small, i.e.,
\begin{equation} \label{eq:efficiency}
    \sum_{j \in [n]} t_j \leq \Order(k). 
\end{equation}

With the approach outlined before we have reduced the whole problem to finding column budgets $t_1, \dots, t_n$ satisfying both conditions~\eqref{eq:accuracy} and~\eqref{eq:efficiency}. This is a natural subtask to which we refer as a ``budget allocation sketch''. Our design of a budget allocation sketch (\cref{thm:budget-alloc}) involves two steps that we outline next; they can be seen as the main innovation of this paper. We are confident that this tool finds more independent applications in the future. As proof of concept, in \cref{sec:distributed} we present a completely different application in a distributed communication scenario.



\subparagraph*{Budget Allocation: Step 1}
Our first step is to estimate $\|X_{-k}\|_F$ in the query setting. This step is surprisingly intricate in light of estimating $\norm{x_{-k}}_2$ for a \emph{vector} $x$ being already known~\cite[Lemma~1.5]{NakosS19}. 
Applied to our setting this merely allows us to estimate $\norm{(X e_j)_{-k}}_2$ for the individual columns of $X$. More generally, we can estimate $\norm{(X e_j)_{-t}}_2$ for all powers $t = 2^0, 2^1, \dots$. Here, by \emph{estimating} we mean a bi-criteria approximation $v_{j, t}$ with constant estimation error as well as a constant loss in the threshold parameter $t$:
\begin{equation*}
    \Omega(1) \cdot \norm{(X e_j)_{-\Order(t)}}_2^2 \leq v_{j, t} \leq \norm{(X e_j)_{-t}}_2^2.
\end{equation*}
Our insight is that, given these approximations $v_{j, t}$, we can read off an estimate (more precisely, a bi-criteria approximation) of $\norm{X_{-k}}_F$ as the optimal value of the following optimization problem: Select, for each column $j$, a value $t_j$ (as a power of $2$) in order to
\begin{alignat*}{2}
    &\text{minimize}\, &&\sum_{j \in [n]} v_{j, t_j} \\
    &\text{subject to}\, &&\sum_{j \in [n]} t_j \leq 2k.
\end{alignat*}
For one direction of the claimed statement, it is easy to see that plugging in $t_j = k_j$ (rounded up to a power of $2$) satisfies the constraint and therefore $\norm{X_{-k}}_F$ provides a valid upper bound on the optimal value. For the other direction, we refer the reader to the formal proof of \cref{lem:frobenius-tail-approx}. Crucially, even solving the above optimization problem \emph{approximately} suffices for our goal.

Computationally, it remains to solve this optimization problem. Our observation is that this problem can be cast as an instance of the \emph{Minimum-Cost Multiple-Choice Knapsack} problem. Here an instance consists of $n$ \emph{items} each associated with a \emph{cost $c_i$} and a \emph{weight $w_i$}. Moreover, the items are partitioned into \emph{bundles.} The goal is to select a subset $I \subseteq [n]$ of items---exactly one item from each bundle---so that the total weight is $\sum_{i \in I} w_i \leq W$ (for some given weight threshold $W$) and so that the total cost $\sum_{i \in I} c_i$ is minimized. This problem has not been explicitly considered before (to the best of our knowledge), but is closely related to the classical \emph{0-1-Knapsack} problem (where the goal is to \emph{maximize} the total \emph{profit} subject to a weight constraint) and the \emph{Multiple Choice Knapsack} problem (where items are similarly grouped into bundles). While both of these problems are NP-hard, it is known that they can be $(1+\epsilon)$-approximated in near-linear time $\widetilde\Order(n \poly(\epsilon^{-1}))$, as was shown for 0-1-Knapsack by Ibarra and Kim~\cite{IbarraK75} and for Multiple Choice Knapsack by Rhee~\cite{Rhee15}. A long line of research has further optimized the polynomial dependence on~$\epsilon$ for 0\=/1\=/Knapsack~\cite{Lawler79,KellererP04,Rhee15,Chan18,Jin19,DengJM23,Mao24,ChenLMZ24b}. We employ Chan's framework for 0-1-Knapsack~\cite{Chan18}, as it transfers nicely also to the Minimum-Cost version of the problem (with some minor tweaks as described in \cref{sec:knapsack} due to space constraints). The consequence is that we can $(1 + \epsilon)$-approximate Minimum-Cost Multiple-Choice Knapsack in near-linear time, and thereby approximate $\norm{X_{-k}}_F$ as desired.\footnote{In light of this surprisingly involved algorithm to approximate $\norm{X_{-k}}_F$, the critical reader might wonder whether we can simplify here. A first instinct could be to not approximate $\norm{X_{-k}}_F$ at all, but instead to spend a logarithmic number of guesses for $\norm{X_{-k}}_F$ (up to a constant factor). Unfortunately, this approach has two drawbacks: First, this approach would lead to a running time overhead of $\log(n\Delta)$ (where $\Delta$ is the largest entry in the given matrices) which we can avoid with the method outlined above (in other words, the approximation via Minimum-Cost Multiple Choice Knapsack has \emph{strongly} polynomial running time). Second, even if we are willing to take this overhead into account, the more serious problem for our application to matrix multiplication is to verify which of the guesses was valid in the end---that is, we would have to verify whether~\smash{$\norm{X - X'}_\infty \leq \frac{1}{\sqrt k} \cdot \norm{X_{-k}}_F$}. This appears to be a quite nontrivial problem and we see no approach other than estimating $\norm{X_{-k}}_F$.}

\subparagraph*{Budget Allocation: Step 2}
The second step is to select the column budgets $t_1, \dots, t_n$. In fact, simply taking values $t_1, \dots, t_n$ of an (approximately) optimal solution to the previous optimization problem already gives a nontrivial budget allocation that will eventually lead to an $\ell_2/\ell_2$ guarantee for our query problem (i.e., we could compute a matrix $X'$ such that $\norm{X - X'}_F \leq (1 + \epsilon) \cdot \norm{X_{-k}}_F$). However, it turns out that with little extra effort we can do better and find budgets that lead to the stronger $\ell_\infty / \ell_2$-guarantee: For each column $j$ select the budget $t_j$ to be the smallest value $t$ satisfying that
\begin{equation*}
    \frac{v_{j, t}}{t} \leq \frac{w}{k},
\end{equation*}
where $w \approx \norm{X_{-k}}_F$ is the approximation computed in step 1. In this way we make sure that no column receives a prohibitively small budget, and it follows from some calculations that the total budget is still bounded by $\Order(k)$ (though with a larger constant); see \cref{lem:col-budgets} for more details.
Notably, the budget allocation routine runs altogether in \emph{near-linear} time in the input sparsity. 


\subsection{Outline}
We start with some preliminaries in \cref{sec:prelims}. In \cref{sec:budget-alloc} we give the formal proof of our key budget allocation step, deferring the discussion of the knapsack-type problem to \cref{sec:knapsack}. We then conclude our the robustifying reduction in \cref{sec:robustifying}. In \cref{sec:sparse-recovery} we give some missing proofs for tools from sparse recovery. Finally, in \cref{sec:distributed} we present an interesting alternative application of our new budget allocation sketch for a distributed heavy hitter problem.
\section{Preliminaries} \label{sec:prelims}
We set $[n] = \set{1, \dots, n}$, and write $\polylog(n) = (\log n)^{\Order(1)}$ and \smash{$\widetilde\Order(n) = n \cdot \polylog(n)$}.

\subparagraph*{Matrix Notation}
Let $A \in \Real^{n \times n}$ be a matrix. We write $A[i, j]$ for the entry at position $(i, j)$. The \emph{Frobenius} norm is defined as $\norm{A}_F = (\sum_{i, j} A[i, j]^2)^{1/2}$, and the $\ell_\infty$-norm is defined as $\norm{A}_\infty = \max_{i, j} |A[i, j]|$. We write $A_{-k}$ for the matrix obtained from $A$ after zeroing out the largest $k$ entries in absolute value (breaking ties in an arbitrary but consistent way). We say that $A$ is \emph{$s$-sparse} if it contains at most~$s$ nonzero entries, and we say that $A$ is \emph{$s$-column sparse} if each column of $A$ contains at most~$s$ nonzero entries.

\subparagraph*{Machine Model}
We work in the standard WordRAM model with word size $\Theta(\log n)$ (where~$n$ is the dimension). To simplify notation, throughout the paper we implicitly assume that all vectors $x \in \Real^n$ and matrices $A \in \Real^{m \times n}$ have entries that can be represented in a $O(\log n)$-bit fixed-precision format; the same applies to complex-valued vectors and matrices. This means that each entry can be stored in a machine word (or in a constant number of machine words). 

Alternatively, our results also hold on the RealRAM model, where a machine word can store an arbitrary real number; in this case vectors and matrices can have arbitrary real entries.

\subparagraph*{Randomization}
We say that an event happens \emph{with high probability} if it happens with probability $1 - 1/n^c$ for an arbitrarily large prespecified constant $c$. Unless stated otherwise, all algorithms are Monte Carlo algorithms that succeed with high probability.

\subsection{Tools from Sparse Recovery}
Throughout we make extensive use of techniques from sparse recovery. Specifically, we will rely on the following two lemmas; the proofs are deferred to \cref{sec:sparse-recovery}. The following lemma is basically the same as~\cite[Lemma 1.4]{NakosS19}. The only difference is that we use random signs instead of Gaussian random variables used in that proof in order to achieve anti-concentration needed for the left-hand side of the inequality.

\begin{restatable}[$\ell_2$-Tail Estimation]{lemma}{lemltwotailapprox} \label{lem:l2-tail-approx}
Let $n \geq 1, n \ge s \geq 0$ be integers and let $\delta > 0$. In time and space $O(\log(1/\delta))$ we can sample a matrix $S^{(s)} \in \set{-1,0,1}^{m \times n}$ with $m = \Order(\log(1/\delta))$ rows and the following properties. (i) Any entry of $S^{(s)}$ can be computed in time $O(1)$. (ii) For any $x \in \Real^n$, given $S^{(s)} x$ we can compute in time $O(\log (1/\delta))$ a value $v$ which satisfies with probability at least $1 - \delta$ that 
\begin{equation*}
    \frac{1}{64} \cdot \norm{x_{-512 s}}_2^2 \leq v \leq \norm{x_{-s}}_2^2.
\end{equation*}
\end{restatable}

The following lemma also follows from~\cite{KaneNPW11,Pagh13,LarsenNNT16}. In \cref{sec:sparse-recovery} we give an accessible self-contained proof which is close in spirit to the one in~\cite{Pagh13}.

\begin{restatable}[Fast Heavy-Hitter Recovery]{lemma}{lemheavyhitters} \label{lem:heavy-hitters}
Let $n \geq 1, s \geq 0$ be integers and let $\delta > 0$. In time and space $O(\log n \log(s /\delta))$ we can sample a matrix $H^{(s)} \in \{-1,0,1\}^{m \times n}$ with $m = \Order(s \log n \log(s/\delta))$ rows and the following properties. (i) The column sparsity of $H^{(s)}$ is $O(\log n\log(s / \delta))$, and any column of the $H^{(s)}$ can be computed in time $O(\log n \log(s / \delta))$. (ii) For any $x \in \Real^n$, given $H^{(s)} x$ we can compute in time $\Order(m)$ an $\Order(s)$-sparse vector $x'$ which satisfies with probability at least $1 - \delta$ that
\begin{equation*}
    \norm{x - x'}_\infty \leq \frac{1}{\sqrt{s+1}} \cdot \norm{x_{-s}}_2.
\end{equation*}
\end{restatable}
\section{Budget Allocation Sketch} \label{sec:budget-alloc}

The goal of this section is to prove the following \cref{thm:budget-alloc} which captures the ``budget allocation sketch'' outlined before. This theorem forms the main technical ingredient for our robustifying reduction in \cref{sec:robustifying}. However, we expect \cref{thm:budget-alloc} to have broader applicability, see for instance \cref{sec:distributed} for an alternative application in a distributed setup.

\begin{theorem}[Budget Allocation Sketch] \label{thm:budget-alloc}
There is a randomized construction of $S \in \mathbb{R}^{s \times n}$ with $s = O(\log^2 n)$, such that given $k \geq 0$ and $SX$ for an unknown $X \in \mathbb{R}^{n \times d}$ with $d \leq n$ we can find in time $O(d \log^4 n)$ for each column $j$ a number $t_j \ge 0$ such that with high probability:
\begin{enumerate} 
\item $\sum_{j \in [d]} t_j = O(k)$,
\item $\frac{\|(Xe_j)_{-t_j}\|_2^2}{t_j+1} \leq \frac{\|X_{-k}\|_F^2}{k+1}$ for each $j \in [d]$.
\end{enumerate}
Furthermore, the matrix $S$ can be sampled in time $O(\log^2 n)$, after which any entry can be computed in time $O(1)$, and it can be stored in $O(\log^2 n)$ space.
\end{theorem}

The construction of $S$ and the proof of properties in the last sentence are simple, see Sections~\ref{sec:budget-alloc:sec:construction} and~\ref{sec:budget-alloc:sec:easy}.
For the main property, we first show that from $SX$ we can accurately and efficiently approximate the total noise $\norm{X_{-k}}_F$ in \cref{sec:budget-alloc:sec:approx-noise}. Then in \cref{sec:budget-alloc:sec:budgets} we show how to allocate column budgets, thus finishing the proof of \cref{thm:budget-alloc}.

\subsection{Construction of the Sketch}
\label{sec:budget-alloc:sec:construction}

We write $T_n = \set{0} \cup \set{2^\ell : 0 \leq \ell \leq \floor{\log n}} \cup \set{n}$. This set has the property that for all $0 \leq x \leq n$ (including $x = 0$), there is a number $t \in T_n$ such that $x \leq t \leq \min\set{2x, n}$.

For each $t \in T_n$ we apply \cref{lem:l2-tail-approx} (with error parameter $\delta = n^{-100}$) to construct the matrix~$S^{(t)}$. We construct the matrix $S$ by stacking all matrices $S^{(t)}$, $t \in T_n$, on top of each other. 

\subsection{Easy Properties}
\label{sec:budget-alloc:sec:easy}

Most properties of Theorem~\ref{thm:budget-alloc} can be easily verified: By \cref{lem:l2-tail-approx} each matrix $S^{(t)}$ has $O(\log 1/\delta) = O(\log n)$ rows and can be sampled in time and space $O(\log n)$, so $S$ has $O(\log^2 n)$ rows and can be sampled in time and space $O(\log^2 n)$. By \cref{lem:l2-tail-approx} each entry of $S^{(t)}$ can be computed in time $O(1)$, so the same holds for $S$. It remains to construct the budgets $t_1,\ldots,t_d$.

\subsection{Approximating the \texorpdfstring{\boldmath$\ell_2$}{l2}-Tail} \label{sec:budget-alloc:sec:approx-noise}
We first show that from the sketch $SX$ we can approximate the $\ell_2$-mass of all columns of $X$ after zeroing out the $t$ largest entries, for any $t \in T_n$:

\begin{lemma}[Column-Wise $\ell_2$-Tail Estimation]
\label{lem:col-l2-tail-approx}
Let $X \in \mathbb{R}^{n \times d}$ with $d \le n$. Given $SX$ in time $O(d \log^2 n)$ we can compute numbers $v_{j,t}$, for all $j \in [d]$ and $t \in T_n$, that satisfy with high probability in~$n$:
\begin{equation*}
    \frac{1}{64} \norm{(X e_j)_{-512 t}}_2^2 \leq v_{j, t} \leq \norm{(X e_j)_{-t}}_2^2.
\end{equation*}
\end{lemma}
\begin{proof}
Fix $j \in [d]$ and $t \in T_n$.
By construction of $S$, from $SX$ we can read off $S^{(t)} X$. The $j$-column of this matrix is $S^{(t)} (X e_j)$. Applying \cref{lem:l2-tail-approx} yields a number $v = v_{j,t}$ with the claimed property
\[ \frac 1{64} \norm{(X e_j)_{-512 t}}_2^2 \le v_{j,t} \le \norm{(X e_j)_{-t}}_2^2. \]
Since one application of \cref{lem:l2-tail-approx} takes time $O(\log 1/\delta) = O(\log n)$ and we apply it $d \log n$ times (once for each $j \in [d]$ and $t \in T_n$), the total time is $O(d \log^2 n)$.
\end{proof}

Next, we combine these approximations to estimate $\norm{X_{-k}}_F$:

\begin{lemma}[Frobenius Norm Tail Estimation] \label{lem:frobenius-tail-approx}
Let $X \in \mathbb{R}^{n \times d}$. Given $SX$ and a number $k \ge 0$ in time $O(d \log^4 n)$ we can compute a number $w$ satisfying with high probability in~$n$:
\begin{equation*}
    \frac{1}{128} \norm{X_{-1024k}}_F^2 \leq w \leq \norm{X_{-k}}_F^2.
\end{equation*}
\end{lemma}
\begin{proof}
We first compute the column-wise approximations $v_{j, t}$ for all $j \in [d]$ and $t \in T_n$ by \cref{lem:col-l2-tail-approx}. The idea is that we can read off $w$ as the optimal value of the following Minimum-Cost Multiple-Choice Knapsack instance: View each pair~\makebox{$(j, t) \in [d] \times T_n$} as an \emph{item} with \emph{cost} $v_{j, t}$ and \emph{weight}~$t$. We group these items into $d$ \emph{bundles} of the form $\set{(j, t) : t \in T_n}$. The goal is to select from each bundle $j \in [d]$ exactly \emph{one} item~$(j, t_j)$ such that the total weight of all items is at most~\makebox{$W = 2k$}, and such that the total cost is minimized. Formally,
\begin{alignat*}{2}
    &\text{minimize}\, &&\sum_{j \in [d]} v_{j, t_j} \\
    &\text{subject to}\, &&\sum_{j \in [d]} t_j \leq 2k.
\end{alignat*}
Let $\OPT$ denote the optimal value of this optimization problem. As we will describe in detail in \cref{sec:knapsack}, we can efficiently compute an approximation $w$ satisfying~\makebox{$\frac12 \OPT \leq w \leq \OPT$} (specifically, we apply \cref{thm:knapsack} in the upcoming \cref{sec:knapsack} with parameter $\epsilon = \frac12$).

It remains to prove that this value $w$ is as desired. To this end, the following claim relates $\OPT$ to the $\ell_2$-mass of the light entries in $X$:

\begin{claim} \label{clm:noise-opt}
$\frac{1}{64} \norm{X_{-1024 k}}_2^2 \leq \OPT \leq \norm{X_{-k}}_F^2$.
\end{claim}
\begin{proof}
We start with the upper bound on $\OPT$---that is, we construct a feasible solution to the optimization problem with value at most $\norm{X_{-k}}_F^2$. Let us call the $k$ largest entries in $X$ (in absolute value) the \emph{heavy} entries. Let $k_j$ denote the number of heavy entries in the $j$-th column. We pick $t_j$ to be the smallest value in $T_n$ that exceeds $k_j$; by the construction of $T_n$ we have $k_j \leq t_j \leq 2k_j$. The constraint $\sum_{j \in [d]} t_j \leq \sum_{j \in [d]} 2 k_j \leq 2k$ is clearly satisfied. Therefore, we have that
\begin{align*}
    \OPT
    &\leq \sum_{j \in [d]} v_{j, t_j} \\
    &\leq \sum_{j \in [d]} \norm{(X e_j)_{-t_j}}_2^2 \\
    &\leq \sum_{j \in [d]} \norm{(X e_j)_{-k_j}}_2^2 \\
    &= \norm{X_{-k}}_F^2,
\end{align*}
where in the second step we have applied \cref{lem:col-l2-tail-approx}.

Next we show the lower bound on $\OPT$. To this end, let $(t^*_1,\ldots,t^*_d)$ denote an optimal solution to the optimization problem. By \cref{lem:col-l2-tail-approx} we have that
\begin{align*}
    \OPT
    &= \sum_{j \in [d]} v_{j, t^*_j} \\
    &\geq \tfrac{1}{64} \sum_{j \in [d]} \norm{(X e_j)_{-512 t^*_j}}_2^2.
\end{align*}
This sum can be seen as the Frobenius norm of the matrix $X$ after zeroing out some~\smash{$512 t^*_j$} entries from each column $j$. We zero out~\smash{$\sum_{j \in [d]} 512 t^*_j \leq 1024k$} entries in total, using that $(t^*_1,\ldots,t^*_d)$ is a feasible solution to the optimization problem. Therefore,~\smash{$\OPT \geq \frac{1}{64} \norm{X_{-1024}}_F^2$} as claimed.
\end{proof}

We finally analyze the running time. 
The call to \cref{lem:col-l2-tail-approx} takes time $\Order(d \log^2 n)$. By \cref{thm:knapsack} it takes time $\Order(d \log^4 n)$ to approximate the Minimum-Cost Multiple-Choice Knapsack problem on $d$ bundles and $d \cdot |T_n| = \Order(d \log n) \le \Order(n^2)$ items for $\epsilon = 1/2$.
\end{proof}

\subsection{Allocating Column Budgets} \label{sec:budget-alloc:sec:budgets}
Finally, we assign \emph{budgets} $t_1, \dots, t_d$ to the columns with the following two objectives in mind: On the one hand the budgets must be large enough such that the $\ell_2$-tails per column, $\norm{(X e_j)_{-t_j}}_2$, are sufficiently small. On the other hand, for efficiency reasons the total budget should be bounded by $\Order(k)$. Compared to the overview, we remark that our formal statements here often involve quantities $\frac{1}{t + 1}$ rather than $\frac{1}{t}$; this is necessary to take care of zero-budget columns.

\begin{lemma}[Column Budgets] \label{lem:col-budgets}
Let $X \in \mathbb{R}^{n \times d}$ with $d \le n$. Given $SX$ and a number $k \ge 0$ in time $O(d \log^4 n)$ we can compute budgets $0 \leq t_1, \dots, t_n \leq n$ that satisfy with high probability in~$n$:
\begin{equation} \label{lem:col-budgets:eq:l2-tail}
    \frac{\norm{(X e_j)_{-t_j}}_2^2}{t_j + 1} \leq \frac{\norm{X_{-k}}_F^2}{k + 1},
\end{equation}
for all $j \in [d]$, and
\begin{equation} \label{lem:col-budgets:eq:total}
    \sum_{j \in [d]} t_j \leq 2^{25} (k + 1) = \Order(k).
\end{equation}
\end{lemma}
\begin{proof}
Compute the approximations $v_{j, t}$ (for~\makebox{$j \in [d]$} and $t \in T_n$) and $w$ by \cref{lem:col-l2-tail-approx,lem:frobenius-tail-approx}. For each $j \in [d]$, let $t_j' \in T_n$ denote the smallest value satisfying that
\begin{equation*}
    \frac{64 v_{j, t_j'}}{t_j' + 1} \leq \frac{w}{k + 1};
\end{equation*}
note that $v_{j, n} = 0$ and thus there certainly is a value $t_j' \in T_n$ satisfying this inequality. Then we pick~\makebox{$t_j = \min\set{512t_j', n}$}. In what follows we argue that properties~\eqref{lem:col-budgets:eq:l2-tail} and~\eqref{lem:col-budgets:eq:total} are satisfied.

Property~\eqref{lem:col-budgets:eq:l2-tail} follows from a simple calculation, using that $v_{j, \ell}$ and $w$ are accurate approximations according to \cref{lem:col-l2-tail-approx,lem:frobenius-tail-approx}:
\begin{align*}
    \frac{\norm{(X e_j)_{-t_j}}_2^2}{t_j+1}
    &\leq \frac{\norm{(X e_j)_{-512 t_j'}}_2^2}{t_j'+1} \\
    &\leq \frac{64 v_{j, t_j'}}{t_j'+1} \\
    &\leq \frac{w}{k+1} \\
    &\leq \frac{\norm{X_{-k}}_F^2}{k+1}.
\end{align*}

For Property~\eqref{lem:col-budgets:eq:total} we define some auxiliary quantities. Let $J \subseteq [d]$ denote the set of indices~$j$ satisfying that $t_j \neq 0$. Note that in the sum \smash{$\sum_{j \in [d]} t_j$} we can ignore all terms with $j \not\in J$. Let~\makebox{$k' = 1024 k$}; we call the $k'$ largest entries in $X$ the \emph{$k'$-heavy} entries. Let $k'_j$ denote the number of $k'$-heavy entries in the $j$-th column of $X$. Then, for each $j \in J$, let $s_j \in T_n$ be the smallest value in $T_n$ which is at least
\begin{equation*}
    \max\set*{k'_j,\, 2^{13} (k+1) \cdot \frac{\norm{(X e_j)_{-k'_j}}_2^2}{\norm{X_{-k'}}_F^2}}.
\end{equation*}
We emphasize that it is not our intention to compute $s_j$ (as we have no way of learning $k'_j$); we merely define $s_j$ for the analysis. We observe that $s_j > 0$ for all $j \in J$, as otherwise we would have~\makebox{$k_j' = 0$} and $\norm{X e_j}_2 = \norm{(X e_j)_{-k_j'}}_2 = 0$, and we would thus have selected $t_j = 0$. Therefore, and by the definition of $T_n$, it follows that
\begin{equation*}
    \max\set*{k'_j,\, 2^{13} (k+1) \cdot \frac{\norm{(X e_j)_{-k'_j}}_2^2}{\norm{X_{-k'}}_F^2}} \leq s_j \leq 2 \cdot \max\set*{k'_j,\, 2^{13} (k+1) \cdot \frac{\norm{(X e_j)_{-k'_j}}_2^2}{\norm{X_{-k'}}_F^2}}.
\end{equation*}
Next, we argue that $t_j' \leq s_j$. Consider that
\begin{align*}
    \frac{64 v_{j, s_j}}{s_j + 1}
    &\leq \frac{64 \cdot \norm{(X e_j)_{-s_j}}_2^2}{s_j} \\
    &\leq \frac{64 \cdot \norm{(X e_j)_{-s_j}}_2^2 \cdot \norm{X_{-k'}}_F^2}{2^{13} (k+1) \cdot \norm{(X e_j)_{-k'_j}}_2^2} \\
    &\leq \frac{\norm{X_{-k'}}_F^2}{128 (k+1)} \\
    &\leq \frac{w}{k+1},
\end{align*}
where have again applied \cref{lem:col-l2-tail-approx,lem:frobenius-tail-approx}. Consequently, from the minimality of the values $t_j'$ it indeed follows that $t_j' \leq s_j$. Finally, we have that
\begin{align*}
    \sum_{j \in [d]} t_j
    &= \sum_{j \in J} t_j \\
    &\leq 512 \cdot \sum_{j \in J} t_j' \\
    &\leq 512 \cdot \sum_{j \in J} s_j \\
    &\leq 1024 \cdot \sum_{j \in J} \max\set*{k'_j, 2^{13} (k+1) \cdot \frac{\norm{(X e_j)_{-k'_j}}_2^2}{\norm{X_{-k'}}_F^2}} \\
    &\leq 1024 \cdot \sum_{j \in J} \left( k'_j + 2^{13} (k+1) \cdot \frac{\norm{(X e_j)_{-k'_j}}_2^2}{\norm{X_{-k'}}_F^2} \right) \\
    &\leq 1024 \cdot (2k' + 2^{14} (k+1)) \\
    &\leq 1024 \cdot (2^{11} k + 2^{14} (k+1)) \\
    &\leq 2^{25} (k + 1),
\end{align*}
proving Property~\eqref{lem:col-budgets:eq:total}.

Finally, consider the running time. The dominant contribution is the calls to \cref{lem:col-l2-tail-approx,lem:frobenius-tail-approx} in time \smash{$\Order(d \log^4 n)$}. Then, having access to~$v_{j, t}$ and $w$, we can easily compute the budgets~$t_1, \dots, t_n$ in time $\Order(d \log n)$.
\end{proof}

This completes the proof of \cref{thm:budget-alloc}, by combining \cref{sec:budget-alloc:sec:easy} and \cref{lem:col-budgets}. Due to space constraints we moved the proof of our Main \cref{thm:robustifying} to \cref{sec:robustifying}.
\section{Robustifying Reduction} \label{sec:robustifying}
The goal of this section is to prove our main theorem:

\thmrobustifying*

\begin{proof}
The first step is to compute the column budgets $s_1, \dots, s_n$ by \cref{thm:budget-alloc}. Specifically, we apply \cref{thm:budget-alloc} for the matrix $X = A B$, and compute the product $S X$ by first multiplying~$S$ with~$A$, and multiplying the resulting $O(\log^2 n) \times n$ matrix with $B$. Let $t_1, \dots, t_n$ denote the resulting column budgets satisfying the two properties
\medskip
\begin{enumerate} 
\item $\sum_{j \in [d]} t_j = O(k)$,
\item $\frac{\|(Xe_j)_{-t_j}\|_2^2}{t_j+1} \leq \frac{\|X_{-k}\|_F^2}{k+1}$ for each $j \in [d]$.
\end{enumerate}
\medskip
We may round up each $t_j$ to the smallest value in $T_n$ while maintaining the two properties.

According to these budgets we partition the columns $[n]$ into blocks $J_s = \set{ j : t_j = s }$ for $s \in T_n$. Let $B^{(s)}$ denote the matrix $B$ restricted to the columns in~$J_s$. In this terminology, our remaining goal is to compute the individual matrix products $A B^{(s)}$ for each $s \in T_n$.

Fix any $s \in T_n$ and let $H^{(s)}$ be the matrix obtained from \cref{lem:heavy-hitters} (with~\makebox{$\delta = n^{-100}$}). We compute the matrix product~$A^{(s)} := H^{(s)} A$ in time $\Order(\IN \log^2 n)$ exploiting that $H^{(s)}$ has column sparsity $\Order(\log^2 n)$ by \cref{lem:heavy-hitters}. Moreover, $A^{(s)}$ has size~\makebox{$\Order(s \log^2 n) \times n$} and is~$\Order(\IN \log^2 n)$-sparse. Therefore, for some constant $\gamma > 0$ to be determined later, we can partition $A^{(s)}$ row-wise into~\makebox{$R = \Order(\log^2 n \cdot \gamma^{-1})$} submatrices~\smash{$A^{(s)}_1, \dots, A^{(s)}_R$} each of which has size $\gamma s \times n$ and is $\IN$-sparse. Indeed, first split $A^{(s)}$ into chunks of $\gamma s$ rows, and whenever necessary subdivide a chunk into smaller subchunks with sparsity at most $\IN$; it is easy to see that the total number of chunks and subchunks does not exceed $\Order(\log^2 n \cdot \gamma^{-1})$. We then compute the matrix products~\smash{$A^{(s)}_1 B^{(s)}, \dots, A^{(s)}_R B^{(s)}$} using oracle calls to the efficient sparse matrix multiplication algorithm. Combining the rows of these resulting matrices appropriately, we obtain $A^{(s)} B^{(s)}$. Finally, we pass the column vectors of~\smash{$A^{(s)} B^{(s)} = H^{(s)} (A B^{(s)})$} to \cref{lem:heavy-hitters} and obtain, for each~\makebox{$j \in J_s$}, a vector $c_j$ such that
\begin{equation*}
    \norm{A b_j - c_j}_\infty \leq \frac{1}{\sqrt{s + 1}} \norm{(A b_j)_{-s}}_2.
\end{equation*}
After all $\Order(\log n)$ iterations, let $C$ denote the $n \times n$ matrix with column vectors $c_1, \dots, c_n$. We return this matrix $C$.

The correctness follows from a simple calculation:
\begin{align*}
    \norm{A B - C}_\infty
    &= \max_{j \in [n]} \norm{A b_j - c_j}_\infty \\
    &\leq \max_{j \in [n]} \frac{1}{\sqrt{t_j + 1}} \norm{(A b_j)_{-t_j}}_2 \\
    &\leq \frac{1}{\sqrt{k+1}} \norm{(A B)_{-k}}_F.
\end{align*}
In the last step we have applied Property~\eqref{lem:col-budgets:eq:l2-tail} from \cref{lem:col-budgets}.

We finally analyze the running time. The construction of the sketch matrix $S$ takes polylogarithmic time, and the multiplications $S A$ and $(S A) B$ take time $O(\IN \log^2 n)$ (since $S$ has $O(\log^2 n)$ rows). Computing the budgets takes time $O(n \log^4 n)$ by \cref{thm:budget-alloc}. For any fixed $s \in T_n$ the time to compute the matrix $H^{(s)}$ is negligible. Computing the matrix product~\smash{$A^{(s)} = H^{(s)} A$} by the naive algorithm takes time $\Order(\IN \log^2 n)$ given that $H^{(s)}$ is $\Order(\log^2 n)$-column sparse. Splitting $A^{(s)}$ into its submatrices runs in the same time $\Order(\IN \log^2 n)$. To bound the time to compute any matrix product~\smash{$A^{(s)}_r B^{(s)}$}, observe that this is a matrix of size $s \times |J_s|$. Since~\smash{$\sum_{j \in [n]} t_j \leq 8 c_2 (c_1^2 + c_2) k$} by Property~\eqref{lem:col-budgets:eq:total} of \cref{lem:col-budgets}, we must have~\smash{$|J_s| \leq 8 c_2 (c_1^2 + c_2) k / s$}. In particular, the matrix~\smash{$A^{(s)}_r B^{(s)}$} has at most~\smash{$\gamma s \cdot 8 c_2 (c_1^2 + c_2) k / s$} nonzero entries, and thus, by choosing the constant $\gamma = 1 / (8 c_2 (c_1^2 + c_2))$, the number of entries becomes $\leq k$. Since both matrices~\smash{$A^{(s)}_r$} and~\smash{$B^{(s)}$} further have size at most~\makebox{$n \times n$} and at most $\IN$ nonzero entries, the running time of each product is $\Order(T(n, \IN, k))$. Recall that we globally iterate over $\Order(\log^2 n)$ choices for $r$ and $|T_n| = \Order(\log n)$ choices for $s$, so the total time is 
\begin{equation*}
    \Order(n \log^4 n + \IN \log^3 n + \log^3 n \cdot T(n, \IN, k)) = \Order(n \log^4 n + \log^3 n \cdot T(n, \IN, k)),
\end{equation*}
as claimed.
\end{proof}

\subsection{Consequences of the Reduction}
We finally show how the three corollaries follow from our main reduction. For \cref{cor:in-plus-k,cor:almost-optimal} we plug in the Abboud--Bringmann--Fischer--Künnemann algorithm~\cite{AbboudBFK24} into \cref{thm:robustifying} (with a different upper bounds on the running time), whereas for \cref{cor:faster-than-pagh} it suffices to use rectangular matrix multiplication. In fact, even for \cref{cor:in-plus-k,cor:almost-optimal} it is not necessary to use the full power of~\cite{AbboudBFK24}, and we could instead directly apply the heavy/light algorithm á la Yuster--Zwick~\cite{YusterZ05}, generalized by Kaplan--Sharir--Verbin~\cite{KaplanSV06}.

\begin{proof}[Proof of \cref{cor:in-plus-k}]
\cite[Theorem~1.4]{AbboudBFK24} proves that $T(n, \IN, \OUT) = \Order((\IN + \OUT)^{1.346})$. From this, the claim is immediate by \cref{thm:robustifying}.
\end{proof}

\begin{proof}[Proof of \cref{cor:faster-than-pagh}]
Let $\omega(\cdot, \cdot, \cdot)$ denote the exponent of algebraic rectangular matrix multiplication (i.e., such that $\MM(n^a, n^b, n^c) = \Order(n^{\omega(a, b, c) + \epsilon})$ for any $\epsilon > 0$). It is well-known that $\omega(\cdot, \cdot, \cdot)$ is a convex function~\cite{LottiR83,Blaser13}.

The running time of the Abboud--Bringmann--Fischer--Künnemann algorithm~\cite{AbboudBFK24} can be bounded by
\begin{equation*}
    \widetilde\Order\parens*{\IN + \max_{\substack{x, z \leq n\\xz \leq \OUT}} \MM(x, n, z)}.
\end{equation*}
Let $\gamma$ be the constant such that $\OUT = n^{\gamma + \order(1)}$. Then we can bound this running time in terms of $\omega(\cdot, \cdot, \cdot)$ by
\begin{equation*}
    \widetilde\Order\parens*{\IN + \max_{\substack{0 \leq a, c \leq 1}} n^{\omega(a, 1, c) + \epsilon}},
\end{equation*}
for any $\epsilon > 0$. From the convexity of $\omega(\cdot, \cdot, \cdot)$ it follows that the maximum is obtained in one of the corner cases $\omega(1, 1, \gamma - 1)$ or $\omega(\gamma - 1, 1, 1)$. It is further known that $\omega(\cdot, \cdot, \cdot)$ is invariant under permuting its three arguments~\cite{LottiR83,Blaser13}, and thus the running time is $\widetilde\Order(\IN + n^{\omega(1, 1, \gamma - 1) + \epsilon}) = \widetilde\Order(\IN + \MM(n, n, \ceil{\frac{k}{n}}) \cdot n^\epsilon)$.

Given this bound for sparse matrix multiplication, the statement follows immediately from \cref{thm:robustifying}.
\end{proof}

\begin{proof}[Proof of \cref{cor:almost-optimal}]
A simple combination of~\cite[Theorem~1.7]{AbboudBFK24} with~\cite[Lemma~4.8]{AbboudBFK24} and the current state-of-the-art bounds for the dual matrix multiplication constant~$\alpha \geq 0.31389$ yields that $T(n, \IN, \OUT) = \Order(\OUT^{1+\epsilon})$, for any $\epsilon > 0$, provided that~\smash{$\OUT \geq \IN^{1.76110} \geq (\IN)^{1 + \frac{1}{1+\alpha}}$}. The statement now follows immediately from our main reduction in \cref{thm:robustifying}.
\end{proof}
\section{Minimum-Cost Multiple-Choice Knapsack} \label{sec:knapsack}
In this section we design an efficient approximation scheme for the \emph{Minimum-Cost Multiple-Choice Knapsack} problem; see the following \cref{thm:knapsack}. For this problem the input consists of $n$ items with \emph{costs}~$c_1, \dots, c_n \geq 0$ and \emph{weights} $w_1, \dots, w_n \geq 0$ along with a partitioning of~$[n]$ into \emph{bundles}~$S_1, \dots, S_m$ and a weight threshold $W \geq 0$. The goal is to select a subset of items~$I \subseteq [n]$, exactly one from each bundle, such that the total weight is at most $\sum_{i \in I} w_i \leq W$ and such that the total cost $\sum_{i \in I} c_i$ is minimized. Equivalently expressed as an integer program with indicator variables $x_1, \dots, x_n \in \set{0, 1}$, the task is to:
\begin{alignat*}{3}
    &\text{minimize}\, &&\sum_{i \in [n]} c_i x_i, \\
    &\text{subject to}\, &&\sum_{i \in [n]} w_i x_i \leq W, \\
    & &&\sum_{i \in S_j} x_i = 1 &&\forall j \in [m].
\end{alignat*}

\begin{theorem} \label{thm:knapsack}
The Minimum-Cost Multiple-Choice Knapsack problem can be $(1+\epsilon)$-approximated in time $\Order(n \log n + m \epsilon^{-2} \log^4 (n / \epsilon)) \leq \Order(n \epsilon^{-2} \log^4 (n / \epsilon))$.
\end{theorem}

To prove this theorem we borrow heavily from Chan's framework for 0-1 Knapsack~\cite{Chan18}. While the dependence on $\epsilon$ is a secondary concern to us, we have to spend a little extra effort in order to cope with the Minimum-Cost variant of the problem. This extra difficulty is mainly due to the fact that we cannot trivially assume that all input numbers are from a $\poly(n, \epsilon^{-1})$-size range. Instead, we employ the following lemma to precompute a very crude approximation:

\begin{lemma} \label{lem:knapsack-crude}
The Minimum-Cost Multiple-Choice Knapsack problem can be $n$-approximated in time $\Order(n \log n)$.
\end{lemma}
\begin{proof}
In order to compute an $n$-approximation, it suffices to consider the relaxed problem with the objective to minimize $\max_i c_i x_i$ rather than $\sum_i c_i x_i$. Then, denoting the optimal value of this alternative problem by $\OPT'$, we have that $\OPT' \leq \OPT \leq n \cdot \OPT'$. To compute $\OPT'$, let us sort and reorder the items such that $c_1 \geq \dots \geq c_n$. Then we successively \emph{delete} the items $1, \dots, n$. In each step we test whether the remaining instance is still feasible (i.e., in the $k$-th step we test whether there is a subset of $\set{k, \dots, n}$ that satisfies the weight and multiple-choice constraints). When $k^*$ is the last feasible iteration it is not hard to see that $\OPT' = c_{k^*}$.

To implement this strategy efficiently we precompute for each bundle $S_j$ its right-to-left-minima, more precisely we compute the item of minimum weight in $S_j \cap \{k,\ldots,n\}$ for each $k$. Note that these are at most $|S_j|$ distinct items, and we can compute them in sorted order by $k$ via a linear-time scan over $S_j$. 
Following the strategy of iteratively deleting the items $1, \dots, n$, we maintain the smallest possible weight of all solutions:
\begin{equation*}
    w = \sum_{j \in [m]} \min_{i \in S_j} w_i.
\end{equation*}
We can initially sort the items, precompute right-to-left-minima and compute~$w$ in time $\Order(n \log n)$.
Then in each step, when we delete some item~$k$, we update the minimum weight of any item in~$S_j$, by advancing in its sequence of right-to-left-minima in case~$k$ is a right-to-left-minimum in~$S_j$. We update~$w$ accordingly (by subtracting~$w_k$ and adding the new minimum weight in the bundle~$S_j$ if~$k$ happened to be a right-to-left-minimum in~$S_j$). Note that at any point in time we can easily check whether the current instance is feasible by testing whether $w \leq W$. The total time is $\Order(n \log n)$ as claimed.
\end{proof}

By running \cref{lem:knapsack-crude} in a preprocessing step we can compute a number $t \leq \OPT \leq n \cdot t$. Afterwards, we can delete all items with cost more than $n \cdot t$ and replace the cost of any item with cost less than $\frac{\epsilon}{n} t$ by $\frac{\epsilon}{n} t$; the latter transformation increases the total cost of any solution additively by at most $\epsilon t \leq \epsilon \OPT$. Then, by rescaling all costs by $\frac{\epsilon}{n} t$ we can assume that all costs are in the range $\set{0} \cup [1, n^2 / \epsilon]$. Hence, this preprocessing ensures that all input numbers are from a $\poly(n, \epsilon^{-1})$-size range.

\subparagraph*{Preliminaries on Monotone Functions}
Before we continue with the actual algorithm, we take a step back and introduce some terminology as in~\cite{Chan18}. Let $f : \Realnneg \to \Realnneg \cup \set{\infty}$ be a nonincreasing function. We say that $f$ has \emph{complexity} $s$ if it takes $s$ distinct values. In particular, since $f$ is nonincreasing this implies that $f$ is a \emph{staircase function} with $s$ steps. Note that any staircase function can be concisely stored in space $\Order(s)$ by storing for each of the at most $s$ different function values $v$ the point $\min\set{x : f(x) = v}$.

Now consider two nonincreasing functions $f, g : \Realnneg \to \Realnneg \cup \set{\infty}$. Their \emph{$(\min, +)$-convolution} is the function $f \oplus g$ defined by
\begin{equation*}
    (f \oplus g)(z) = \min_{\substack{x, y \in \Realnneg\\x+y=z}} (f(x) + g(y));
\end{equation*}
it can easily be verified that $f \oplus g$ is also nonincreasing. Moreover, if the complexities of $f$ and~$g$ are at most $s$ and $t$, respectively, then the complexity of~$f \oplus g$ is at most $s t$, and we can also compute~$f \oplus g$ in time $\Order(s t)$.

Finally, we say that a function $\widetilde f$ is an \emph{$(1+\epsilon)$-approximation} of another function $f$ if for all points $x \in \Realnneg$ we have $f(x) \leq \widetilde f(x) \leq (1 + \epsilon) f(x)$. Let $\stair_\epsilon(f)$ be the function in which each nonzero non-$\infty$ value $f(x)$ is rounded up to the smallest power of $1 + \epsilon$. By definition, $\stair_\epsilon(f)$ is an $(1+\epsilon)$-approximation of $f$ with complexity at most
\begin{equation*}
    \Order\parens*{\log_{1+\epsilon} \frac{\max_{< \infty}(f)}{\min_{> 0}(f)}} = \Order\parens*{\epsilon^{-1} \log \frac{\max_{< \infty}(f)}{\min_{> 0}(f)}},
\end{equation*}
where we write $\max_{< \infty}(f)$ for the largest non-$\infty$ function value and $\min_{> 0}(f)$ for the smallest positive function value of $f$.

\subparagraph*{A Divide-and-Conquer Algorithm}
For a subset of bundles $J \subseteq [m]$ we define the function $f_J : \Realnneg \to \Realnneg \cup \set{\infty}$ as follows. Let $f(w)$ denote the minimum cost among all solutions with weight at most $w$ that select exactly one item from each bundle in $J$ (and no item from the other bundles). We set $f(w) = \infty$ if there is no such solution.

\begin{observation} \label{obs:weight-fn}
$f_J$ is nonincreasing and $\OPT = f_{[m]}(W)$.
\end{observation}

In light of this observation, in order to solve the Minimum-Cost Multiple-Choice Knapsack problem our goal is to compute the function $f_{[m]}$. In fact, as we are only shooting for a $(1 + \epsilon)$-approximate solution our aim is to compute a $(1 + \epsilon)$-approximation~\smash{$\widetilde f_{[m]}$} of $f_{[m]}$. The simple insights behind the algorithm are that:

\begin{observation} \label{obs:knapsack-conv}
For any partition $J = J_1 \sqcup J_2$ it holds that $f_J = f_{J_1} \oplus f_{J_2}$.
\end{observation}

\begin{observation} \label{obs:conv-approx}
Let $\widetilde f$, $\widetilde g$ be $(1 + \epsilon)$-approximations of $f$, $g$, respectively. Then $\widetilde f \oplus \widetilde g$ is a $(1 + \epsilon)$-approximation of $f \oplus g$.
\end{observation}

\begin{proof}[Proof of \cref{thm:knapsack}]
Let $\delta > 0$ be a parameter. We design a divide-and-conquer algorithm to recursively compute good approximations $\widetilde f_J$ of the functions $f_J$. At the base level, for any $j \in [m]$, we compute the functions~$f_{\set j}(w) = \min\set{c_i : i \in S_j, w_i \leq w}$ exactly and then pick $\widetilde f_{\set j} := \stair_\delta(f_{\set j})$. For the recursive step we partition $J = J_1 \sqcup J_2$ arbitrarily into halves of (roughly) equal size and compute~\smash{$\widetilde f_{J_1}$} and~\smash{$\widetilde f_{J_2}$} recursively. We then compute their $(\min, +)$-convolution~\smash{$\widetilde f_{J_1} \oplus \widetilde f_{J_2}$} and choose~\smash{$\widetilde f_J := \stair_\delta(\widetilde f_{J_1} \oplus \widetilde f_{J_2})$}. In this way we compute an approximation $\widetilde f_{[m]}$ in a binary-tree fashion with recursion depth $\Order(\log m)$.

By induction it is easy to verify that $\widetilde f_J$ is an $(1 + \delta)^{d+1}$-approximation of $f_J$ when $J$ is at the $d$-th level of the recursion (where $d = 0$ is the base level). This is clear at the base level. At any higher level $d$ we recursively obtain $(1 + \delta)^d$-approximations~\smash{$\widetilde f_{J_1}$} and~\smash{$\widetilde f_{J_2}$}, and thus~\smash{$\widetilde f_{J_1} \oplus \widetilde f_{J_2}$} is a $(1 + \delta)^d$-approximation of $f_{J_1} \oplus f_{J_2}$ by \cref{obs:conv-approx}, which in turn equals $f_J$ by \cref{obs:knapsack-conv}. The transformation by $\stair_\delta(\cdot)$ worsens the approximation by another factor $(1 + \delta)$.

At the root of the recursion the error has accumulated to $(1 + \delta)^{\Order(\log m)}$, and therefore by setting~\smash{$\delta = \Theta(\epsilon / \log m)$} we can compute a $(1 + \epsilon)$-approximation of $f_{[m]}$. By \cref{obs:weight-fn} we can read off an $(1 + \epsilon)$-approximation of the given Minimum-Cost Multiple-Choice Knapsack instance.

Regarding the running time, note that all functions $f_J$ and~\smash{$\widetilde f_J$} have range $\set{0} \cup [1, n^2 / \epsilon] \cup \set{\infty}$ due to the preprocessing with \cref{lem:knapsack-crude}, and thus their $\stair_\delta(\cdot)$ approximations have complexity bounded by $s = \Order(\delta^{-1} \log (n / \epsilon))$. At the $j$-th leaf we spend time~$\Order(|S_j| \log |S_j|)$ to sort the weights in the $j$-th bundle and to compute $f_{\set j}$. Then we spend time $\Order(s)$ to compute~\smash{$\widetilde f_{\set j}$}. Any other node in the recursion tree takes time $\Order(s^2)$ to compute the $(\min, +)$-convolution~\smash{$\widetilde f_{J_1} \oplus \widetilde f_{J_2}$} plus time $\Order(s)$ to compute~\smash{$\widetilde f_J$}. As the recursion tree has $\Order(m)$ nodes, the total running time is
\begin{equation*}
    \Order\parens*{\sum_{j \in [m]} |S_j| \log |S_j| + m s^2} = \Order(n \log n + m \epsilon^{-2} \log^2(m) \log^2(n / \epsilon)) = \Order(n \epsilon^{-2} \log^4(n / \epsilon)),
\end{equation*}
as claimed.
\end{proof}

\bibliographystyle{plainurl}
\bibliography{paper}

\appendix
\def\enc{\mathrm{Enc}}
\section{Sparse Recovery Tools} \label{sec:sparse-recovery}
In this section we provide the missing proofs of \cref{lem:l2-tail-approx,lem:heavy-hitters}. We will use the well-known AMS sketch for estimating the $\ell_2$ norm of a vector:


\begin{lemma}[AMS Sketch~\cite{AlonMS99}] \label{lem:ams-sketch}
Let $n \geq 1$ be an integer and let $\delta > 0$. In time and space~$O(\log (1/\delta))$ we can sample a matrix $A \in \{-1, 1\}^{c \times n}$ with $c = \Order(\log(1/\delta))$ rows and the following properties. (i) Any entry of $A$ can be computed in time $O(1)$. (ii) For any vector $x \in \mathbb{R}^n$, given $y = Ax$, we can compute in $\Order(\log(1/\delta))$ time a value $v$ which satisfies \smash{$\frac12\cdot \norm{x}_2^2 \leq v \leq 2 \cdot \norm{x}_2^2$} with probability at least $1-\delta$.
\end{lemma}

The following lemma is basically the same as~\cite[Lemma 1.4]{NakosS19}. Here we present a simplified proof, and we use random signs instead of Gaussian random variables as used in that proof.

\lemltwotailapprox*

\begin{proof}
For $s = 0$ the claim follows from \cref{lem:ams-sketch} after scaling $v$ by a factor $\frac 12$, so assume $s > 0$. We will prove the claim for $m = \Order(1)$ with success probability $\frac23$. The result then follows from a standard boosting argument, by repeating the process $\Theta(\log (1/\delta))$ times and taking the median.

\subparagraph*{Description of the Sketch Matrix.}
We start with a description of the sketch matrix $S = S^{(s)}$.
Sample a set $T \subseteq [n]$ containing each $i \in [n]$ pairwise independently with probability $\frac{1}{16 s}$.\footnote{That is, we sample a pairwise independent hash function $h : [n] \to [16s]$, and let $T = \set{i : h(i) = 0}$, say. It is well-known that such a hash function can be sampled with only $O(\log n)$ bits, i.e., $O(1)$ words, of randomness.} Let~$\Pi_T$ be the projection matrix that zeroes out coordinates outside of~$T$, i.e, $\Pi_T$ is a diagonal matrix whose $i$-th diagonal entry is 1 if $i \in T$ and 0 otherwise. We use the AMS sketch from \cref{lem:ams-sketch} with parameter $\delta = 0.01$ to construct a matrix $A$ and set $S := A \Pi_T$. It is clear that each entry of~$S$ can be computed in time $O(1)$ as each entry of $A$ can be computed in time $O(1)$ by \cref{lem:ams-sketch}.

\subparagraph*{Description of the Query Algorithm.}
Next we describe the algorithm that, given $S x$ for some vector $x \in \Real^n$, returns an approximation $v$ as in the statement. Simply feed $S x = A (\Pi_T x)$ to the algorithm from \cref{lem:ams-sketch} to obtain a value \smash{$\frac12\cdot \norm{\Pi_T x}_2^2 \leq v \leq 2 \cdot \norm{\Pi_T x}_2^2$} (with probability at least $0.99$). Return $\bar v := \frac s2 \cdot v$. Clearly the running time is $O(1)$ by \cref{lem:ams-sketch}.

\subparagraph*{Correctness: Upper Bound.}
We argue that this algorithm is correct with probability at least~\smash{$\frac23$}.
%
%
%
%
Let $\sigma$ be a permutation of $[n]$ satisfying $|x_{\sigma(1)}| \ge \ldots \ge |x_{\sigma(n)}|$. In particular, $x_{\sigma(1)},\ldots,x_{\sigma(s)}$ are the largest $s$ entries of $x$ in absolute value. Let $R := \{\sigma(s+1),\ldots,\sigma(n)\}$ denote the remaining entries. 
Note that for any $i \in R$ the contribution of the $i$-th coordinate to $\norm{\Pi_T x }_2^2$ is $x_i^2$ if $i \in T$, which happens with probability $\frac{1}{16 s}$, and it is 0 otherwise. By linearity of expectation we thus have
\begin{equation*}
    \Ex\big[ \norm{\Pi_{T \cap R} \, x }_2^2 \big] = \frac 1{16 s} \cdot \norm{ \Pi_R \, x }_2^2 = \frac 1{16 s} \cdot \norm{ x_{-s} }_2^2.    
\end{equation*}
From Markov's inequality it follows that $\norm{\Pi_{T \cap R} \cdot x }_2^2 \le \frac 1{s} \cdot \norm{ x_{-s} }_2^2$ with probability at least $1-\frac{1}{16}$.
Moreover, for the $s$ largest entries of $x$, by union bound we have $T \cap \{\sigma(1),\ldots,\sigma(s)\} = \emptyset$ with probability at least $1 - s \cdot \frac 1{16 s} = 1-\frac{1}{16}$. Combining both parts and applying the union bound, we arrive at $\norm{\Pi_T x}_2^2 \le \frac 1{s} \cdot \norm{ x_{-s} }_2^2$ with probability at least $1-\frac{1}{8}$. Recall that the value~$v$ computed by the AMS sketch satisfies \smash{$\frac12\cdot \norm{\Pi_T \cdot x}_2^2 \leq v \leq 2 \cdot \norm{\Pi_T \cdot x}_2^2$} with probability at least $0.99$; we condition on this event in what follows.
Combining this with the previous fact, we obtain the second claimed inequality:
\begin{equation*}
    \bar v = \frac s2 \cdot v \le s \cdot \norm{\Pi_T \cdot x}_2^2 \le \norm{x_{-s}}_2^2.
\end{equation*}

\subparagraph*{Correctness: Lower Bound.}
We continue with the other direction. Let $y \in \Real^n$ be the vector defined by
\begin{equation*}
    y_i = \min\left\{|x_i|^2, \frac{1}{512s} \cdot \norm{x_{-512s}}_2^2\right\}.
\end{equation*}
Moreover, let $\delta_i \in \{0,1\}$ indicate whether $i \in T$. Our first goal is to prove $\norm{y}_1 \geq \norm{x_{-512s}}_2^2$. Focus on an index $i$ with $y_i > |x_i|^2$, or equivalently,
\begin{equation*}
    |x_i|^2 > \frac{1}{512s} \cdot \norm{x_{-512s}}_2^2.
\end{equation*}
Clearly there can be at most $1024s$ such indices: The largest $512s$ coordinates in $x$, plus at most $512s$ more coordinates satisfying this inequality. Let $i^*$ be the largest index with $y_{i^*} < |x_{i^*}|^2$. There are two cases: On the one hand, if $i^* \geq 512s$, then clearly
\begin{equation*}
    |x_{\sigma(1)}|^2 \geq \dots \geq |x_{\sigma(512s)}|^2 > \frac{1}{512s} \cdot \norm{x_{-512s}}_2^2,
\end{equation*}
and thus
\begin{equation*}
    \norm{y}_1 \geq \sum_{i \leq 512s} y_{\sigma(i)} \geq \frac{512s}{512s} \cdot \norm{x_{-512s}}_2^2 = \norm{x_{-512s}}_2^2.
\end{equation*}
On the other hand, if $i^* < 512s$ then
\begin{equation*}
    \norm{y}_1 \geq \sum_{i>512s} y_{\sigma(i)} = \sum_{i>512s} |x_{\sigma(i)}|^2 = \norm{x_{-512s}}_2^2.
\end{equation*}

We now turn our attention to the random variable $\sum_{i=1}^n \delta_i y_i$. Clearly, $|x_i|^2$ majorizes $y_i$ and therefore $\Pr[\sum_{i=1}^n \delta_i y_i \geq \lambda] \leq \Pr[\sum_{i=1}^n \delta_i |x_i|^2 \geq \lambda] = \Pr[\|\Pi_T x\|_2^2 \geq \lambda]$ for any $\lambda$. By linearity of expectation we have
\begin{equation*}
    \Ex\brackets*{\sum_{i=1}^n \delta_i y_i} = \frac{\norm{y}_1}{16 s},
\end{equation*}
and by pairwise independence of the random variables $\delta_i$ we have
\begin{align*}
    \Var\brackets*{\sum_{i=1}^n\delta_iy_i}
    &= \sum_{i=1}^n y_i^2 \Var[\delta_i] \\
    &\leq \sum_{i=1}^n y_i^2 \cdot \frac{1}{16 s}
\intertext{To bound this further we use that $y_i \leq \frac{1}{512s} \cdot \norm{x_{512s}}_2^2$ and that $\norm{x_{-512s}}_2^2 \leq \norm{y}_1$ (as shown before):}
    &\leq \frac{\norm{y}_1}{16 s} \cdot \frac{\norm{x_{-512s}}_2^2}{512s} \\
    &\leq \frac{\norm{y}_1}{16 s} \cdot \frac{\norm{y}_1}{512s} \\
    &= \frac{\norm{y}_1^2}{8192 s}.
\end{align*}
Applying Chebyshev's inequality to the random variable $\sum_{i=1}^n \delta_i y_i$ we obtain that 
\begin{equation*}
    \Pr\brackets*{\abs*{\sum_{i=1}^n y_i \delta_i - \frac{\norm{y}_1}{16 s}} > \frac{\norm{y}_1}{32 s}} \leq \frac{\Var\brackets*{\sum_{i=1}^n\delta_i y_i}}{\parens*{\frac{\norm{y}_1}{32 s}}^2}\leq \frac{1}{8}.
\end{equation*}
This implies that with probability at least $1-\frac{1}{8}$ we have that
\begin{equation*}
    \norm{\Pi_T x}_2^2 \geq \frac{\norm{y}_1}{32 s} \geq \frac{\norm{x_{-512s}}_2^2}{32s}.
\end{equation*}
Combining with the guarantee of the AMS sketch we obtain the desired bound
\begin{equation*}
    \bar v \geq \frac{\norm{\Pi_T x}_2^2}{2} \geq \frac{\norm{x_{-512s}}_2^2}{64s}.
\end{equation*}
The total failure probability is $\frac{1}{8}$ (from the upper bound) plus $\frac{1}{8}$ (from the lower bound) plus $0.01$ (from the AMS sketch) which is $0.26 < \frac{1}{3}$ as claimed.
\end{proof}

Next, we focus on the proof of \cref{lem:heavy-hitters}. We will use the ubiquitous Count-Sketch for estimation of coordinates.

\begin{lemma}[Count-Sketch~\cite{CharikarCF02}]
\label{lem:countsketch}
Let $n \geq 1$ and $n \geq s \geq 0$ be integers and let $\delta > 0$. In time and space $O(\log(1/\delta))$ we can sample a matrix $C^{(s)} \in \set{-1, 0, 1}^{m \times n}$ with $m = O(s \log(1/\delta))$ rows and the following properties. (i) The column sparsity of $C^{(s)}$ is $O(\log(1/\delta))$, and any column of $C^{(s)}$ can be computed in time $O(\log(1/\delta))$. (ii) For any $x \in \mathbb{R}^n$, given $C^{(s)} x$ and $i \in [n]$ we can find in time $O(\log(1/\delta))$ an estimate $x_i'$ which satisfies with probability at least $1 - \delta$ that
\begin{equation*}
    |x_i - x_i'|\leq \frac{1}{\sqrt{s}}\cdot \norm{x_{-s}}_2.
\end{equation*}
\end{lemma}

We will use an efficiently encodable and decodable error-correcting code that corrects a constant fraction of errors, e.g., the following construction due to Spielman~\cite{Spielman96}.

\begin{lemma}[Superconcentrator Error-Correcting Code~\cite{Spielman96}]
\label{lem:spielman}
There exist absolute constants $\alpha < 1 < \beta$ and a function $\enc: \{0,1\}^t \rightarrow \{0,1\}^{\beta t}$ such that $\enc(x)$ can be computed in $O(t)$ time and an $\alpha$ fraction of errors can be corrected in $O(t)$ time, i.e., given $y$ which differs from $\enc(x)$ in at most an $\alpha$ fraction of its coordinates we can recover $x$ in $O(t)$ time.
\end{lemma}

\lemheavyhitters*

\begin{proof}
For $s = 0$ the claim is trivial (by choosing $x'$ as the all-zeros vector), so assume that~\makebox{$s > 0$}. The main effort is to recover a small set $L$ that contains all \emph{heavy hitters}, i.e., all coordinates $i$ with \smash{$|x_i| \geq \frac{1}{\sqrt{s+1}} \norm{x_{-s}}_2$}. We will focus on this subtask for the main part of this proof, and then comment in the final paragraph how to derive the vector $x'$ from $L$.

\subparagraph*{Description of the Sketch Matrix.}
Without loss of generality assume that $n$ is a power of~2, and identify $[n]$ with $\set{0, 1}^{\log n}$. In what follows $\enc : \set{0, 1}^{\log n} \to \set{0, 1}^{\beta \cdot \log n}$ is the error correcting code from \cref{lem:spielman} with the corresponding constants $\alpha$ and $\beta$. Let $\gamma \ge 1$ be a large constant to be determined and let $R = \gamma \log(s/ \delta)$. Moreover, let $A_1, \dots, A_R \in \set{-1, 1}^{c \times n}$ denote independent AMS sketch matrices with error parameter $0.01$ and $c = O(1)$ rows (\cref{lem:ams-sketch}). We also pick pairwise independent hash functions
\begin{alignat*}{2}
    h_r &: [n] \to [\gamma s]\qquad &&\text{(for $r \in [R]$)}, \\
    \tau_{r, \ell} &: [n] \to \set{-1, +1}\qquad &&\text{(for $r \in [R],\, \ell \in [\beta \log n]$)},
\intertext{Now consider the values}
    v_{r, b, j} &= \sum_{\substack{i \in [n]\\h_r(i) = b}} A_r[j, i] \cdot x_i \qquad&&\text{(for $r \in [R],\, b \in [\gamma s],\, j \in [c]$)}, \\
    w_{r, b, \ell} &= \sum_{\substack{i \in [n]\\h_r(i) = b\\\enc(i)[\ell] = 1}} \tau_{r, \ell}(i) \cdot x_i \qquad&&\text{(for $r \in [R],\, b \in [\gamma s],\, \ell \in [\beta \log n]$)}.
    \end{alignat*}
Note that these quantities are linear combinations of the $x_i$'s and can thus be expressed as a matrix-vector product $H^{(s)} x$, where $H^{(s)}$ has  column sparsity $(\beta \log n +\gamma)\cdot R = O(\log n \log (s/\delta))$. Moreover, given any column $i \in [n]$, we can easily report the nonzero entries in the $i$-th column in time $O(\log n \log(s/\delta))$.

\subparagraph*{Description of the Query Algorithm.}
We are given~$H^{(s)} x$ for some (unknown) vector $x$. Our goal is to compute a set $L$ that (1) contains every heavy hitter, and (2) has size at most $|L| = O(s)$.

Fix $r, b$, and let $T := \set{ i : h_r(i) = b}$. Note that from the values~\smash{$(v_{r, b, j})_{j \in [c]}$} we can learn $\norm{x_T}_2^2$ with relative error $\frac12$ with probability at least $0.99$ by the AMS sketch (\cref{lem:ams-sketch}); denote this approximation by $d_{r, b}$. We define a bit-string~\smash{$Z \in \set{0,1}^{\beta \log n}$} by~\smash{$Z[\ell] = 1$} if $|w_{r,b,\ell}| \geq \frac{1}{2} \cdot \sqrt{d_{r, b}}$, and~\smash{$Z[\ell] = 0$} otherwise---this can be implemented by iterating over all $\ell \in [\beta \log n]$ and spending constant time per $\ell$. We subsequently perform error correction on $Z$ using the linear-time decoder associated to $\enc$ to infer an index $i \in \{0,1\}^{\log n}$. Across all pairs $(b, r)$ this yields a total of $\gamma s \cdot R$ indices. We let $L$ be the of all those indices that appear at least $\frac23 \cdot R$ times. 

\subparagraph*{Correctness.}
Clearly, there can be at most $O(s)$ such indices, proving (2). To prove (1), we claim that with probability at least $1 - \delta$, all the heavy hitters are in $L$. Recall that here by heavy hitter we mean a coordinate $i$ with
\begin{equation*}
    |x_i| \geq \frac{1}{\sqrt{s+1}} \cdot \norm{x_{-s}}_2^2.
\end{equation*}
Note that there can be at most~$2s+1$ heavy hitters---the $s$ largest elements in $x$ plus at most $s+1$ more coordinates. Fix a particular heavy hitter $i^\star$, a particular $r^\star \in [R]$ and let $b^\star := h_{r^\star}(i^\star)$. With probability~\smash{$1 - \frac{2s}{\gamma s} = 1 - \frac{2}{\gamma}$} the coordinate $i^\star$ will be isolated from the other heavy hitters under~$h_{r^\star}$. Condition on this event, and consider the expected $\ell_2^2$-mass of the non-heavy hitters in the bucket~$b^\star$:
\begin{equation*}
    \Ex\brackets*{\sum_{\substack{i \in [n], i \neq i^\star\\h_{r^\star}(i)=b^\star}} |x_i|^2} = \frac{\norm{x_{-s}}}{\gamma s}.
\end{equation*}
By Markov's inequality it follows that with probability at least $1 - \frac{200}{\alpha \gamma}$ we have that
\begin{equation*}
    \sum_{\substack{i \in [n], i \neq i^\star\\h_{r^\star}(i)=b^\star}} |x_i|^2 \leq \frac{\alpha}{200s} \norm{x_{-s}}_2^2 \leq \frac{\alpha}{100} \cdot |x_{i^\star}|^2.
\end{equation*}
Conditioning on this event, \cref{lem:ams-sketch} guarantees that the approximation $d_{r^\star, b^\star}$ satisfies that 
\begin{equation*}
    \frac12 \cdot |x_{i^\star}|^2 \leq d_{r^\star, b^\star} \leq 2 \cdot \parens*{|x_{i^\star}|^2 + \frac{\alpha}{100} \cdot |x_{i^\star}|^2} < 2.1 \cdot |x_{i^\star}|^2,
\end{equation*}
with probability at least $0.99$. We now show that in this event, for each fixed $\ell \in [\beta \log n]$, the expected total contribution from the non-heavy hitters to $w_{r^\star,b^\star,\ell}$ is small. Specifically, consider the following random variable, which depends only on the randomness of $\tau$,
\begin{equation*}
    X_{r^\star, b^\star, \ell} := w_{r^\star, b^\star, \ell} - \tau_{r^\star, \ell}(i^\star) \cdot \enc(i^\star)[\ell] \cdot x_{i^\star}.
\end{equation*}
Its expectation is zero (as each term in the sum is multiplied by a random sign $\tau_{r^\star, \ell}(i)$), and its variance can be bounded as follows:
\begin{align*}
    &\Var \brackets*{\,X_{r^\star, b^\star, \ell}\,} = \Var \brackets*{\,w_{r^\star, b^\star, \ell} - \tau_{r^\star, \ell}(i^\star) \cdot \enc(i^\star)[\ell] \cdot x_{i^\star}\,} \\
    &\qquad\leq \Ex \brackets*{\,\parens*{\sum_{\substack{i \in [n], i \neq i^\star\\h_{r^\star}(i) = b^\star\\\enc(i)[\ell] = 1}} \tau_{r^\star, \ell}(i) \cdot x_i}^2\,} \\
    &\qquad\leq \sum_{\substack{i \in [n], i \neq i^\star\\h_{r^\star}(i) = b^\star \\\enc(i)[\ell] = 1}} x_i^2\\
    &\qquad\leq \frac{\alpha}{100} \cdot x_{i^\star}^2,
\end{align*}
Thus, by Chebyshev's inequality it follows that with probability at least $1-\frac{\alpha}{25}$ we have that
\begin{equation*}
    \abs*{w_{r^\star, b^\star, \ell} - \enc(i^\star)[\ell] \cdot x_{i^\star}} \leq 2 \sqrt{\Var[\, X_{r^\star, b^\star, \ell}\,]} \leq \frac14 \cdot |x_{i^\star}|.
\end{equation*}
Now if $\enc(i^\star)[\ell] = 1$ then
\begin{align*}
    |w_{r^\star, b^\star, \ell}| \geq \frac{3}{4} \cdot |x_{i^\star}| \geq \frac{3}{4} \cdot \sqrt{\frac{1}{2.1} \cdot d_{r^\star, b^\star}} > \frac{1}{2} \cdot \sqrt{d_{r^\star, b^\star}}.
\end{align*}
On the other hand, if $\enc(i^\star)[\ell] = 0$ then
\begin{equation*}
    |w_{r^\star, b^\star, \ell}| \leq \frac{1}{4} \cdot |x_{i^\star}| \leq \frac14 \cdot \sqrt{2 d_{r^\star, b^\star}} < \frac12 \cdot \sqrt{d_{r^\star, b^\star}}.
\end{equation*}
Since for each $\ell$ we assign $Z[\ell] = 1$ if and only if $|w_{r^\star, b^\star, \ell}| \geq \frac12 \cdot \sqrt{d_{r^\star, b^\star}}$, it follows that $Z[\ell] = \enc(i^\star)[\ell]$ with probability at least $1 - \frac{\alpha}{25}$. By the Markov's inequality, $Z$ will differ from $\enc(i^\star)$ in less than an $\alpha$-fraction of the coordinates with probability $\frac{24}{25}$. This suffices to ensure that the recovered index from the decoding procedure associated with $C$ equals $i^\star$. All in all, the probability of not recovering $i^\star$ is at most $\frac{2}{\gamma} + \frac{200}{\alpha \cdot \gamma} + 0.01 + \frac{1}{25} < \frac{1}{4}$, for a suitable choice of $\gamma$. The overall $R = \Order(\log(s/\delta))$ repetitions ensure that the procedure fails with probability at most $\delta$ by a union bound over all $O(s)$ heavy hitters.

\subparagraph*{Recovering the Approximation \boldmath$x'$.}
So far we have only given a sketch to recover the set $L$, and it remains to argue how to obtain $x'$. In addition, we set up a Count-Sketch matrix $C^{(s+1)}$ (\cref{lem:countsketch}) with failure probability $\Theta(\delta/s)$, and append this matrix to our sketch matrix $H^{(s)}$ from before. The additional column sparsity is $O(\log(s/\delta))$, the number of rows increases by $O(s \log (s/\delta))$, and we can still construct all nonzero entries in a column in linear time.

Then, in the recovery algorithm we simply query Count-Sketch to obtain approximations $x_i'$, for all $i \in L$, satisfying that
\begin{equation*}
    |x_i - x_i'| \leq \frac{1}{\sqrt{s+1}} \cdot \norm{x_{-s}}_2.
\end{equation*}
Let $x'$ be the vector consisting of these approximations whenever $i \in L$, and zeros elsewhere. By a union bound over the $|L| = O(s)$ list entries, all entries in $L$ is correctly approximated with probability at least $1 - \delta$. Moreover, as argued before, with probability at least $1 - \delta$ all missing entries in $x'$ are not heavy hitters and can therefore be correctly approximated by $0$. The total error probability is $1 - 2\delta$ which yields the claim (by adjusting $\delta$).
\end{proof}
\section{Distributed Heavy Hitters} \label{sec:distributed}
In this section we present another simple application of our new budget allocation sketch (\cref{thm:budget-alloc}) to a distributed/streaming heavy hitters problem.

In the usual turnstile streaming setting, the goal is to maintain a vector $x \in \mathbb{R}^n$, which is initially all-zeroes, under a stream of updates of the form $x_i := x_i + \Delta$, for given $i,\Delta$. The goal is to answer interesting queries about $x$ (approximately and with high probability), while minimizing the space usage. In the multi-pass setting, multiple passes over the stream of updates are allowed, and we want to make as few passes as possible.

We consider the following distributed variant of the usual turnstile multi-pass streaming setting. We have $d$ players and an arbitrary partitioning of the coordinates $[n] = I_1 \cup \ldots \cup I_d$. The stream is split over the players, namely player~$j$ receives exactly the updates affecting the coordinates in~$I_j$. 
In addition, there is a coordinator, which can communicate with the players. 
Again the goal is to answer interesting queries about $x$, while minimizing the number of passes over the stream, the total space used by the players, the number of rounds of communication, and the total amount of communication (as well as the running time used by the players and coordinator).

We study this setting for the heavy hitters problem: The task is that the coordinator knows the heavy hitters of $x$, more precisely the coordinator computes a set $R$ of size $|R| = O(1/\epsilon^2)$ that with high probability contains each index $i$ with $|x_i| \ge \epsilon \cdot \norm{x_{-1/\epsilon^2}}_2$.

We show that this distributed streaming heavy hitters problem can be solved with 2 passes over the streams, $O((\epsilon^{-2} + d) \log^2 n)$ space in total (over all players and the coordinator, in words), 4 communication rounds, $O((\epsilon^{-2} + d) \log^2 n)$ total communication (in words), and total time $O(d \log^4 n + (\epsilon^{-2} + u) \log^2 n)$, where $u$ is the total length of the input streams.

\subparagraph*{The Algorithm}
For notational convenience, we encode the partitioning $I_1 \cup \ldots \cup I_d$ by the matrix $X \in \mathbb{R}^{d \times n}$ with $X[i, j] = x_i$ if $i \in I_j$, and 0 otherwise. Then the $j$-th column $X e_j$ is equal to the vector given by the stream of player $j$.
Observe that computing the heavy hitters of $x$ is equivalent to computing the heavy hitters of $X$. We run the following protocol.

\begin{enumerate}
\item The coordinator samples the matrix $S$ from \cref{thm:budget-alloc} and sends it to all players. (If all players have access to shared randomness, then this communication round is not necessary.)
\item Now the players make one pass over their streams, where player $j$ computes $S (X e_j)$. They send the results to the coordinator, who now knows $SX$.
\item The coordinator now uses \cref{thm:budget-alloc} with $k := \lceil 1/\epsilon^2 \rceil$ to compute budgets $t_1,\ldots,t_d$. She informs each player of their respective assigned budgets $t_j$.
\item Now the players make a second pass over their streams, where player $j$ computes the $t_j$-heavy hitters (i.e., identifies the entries in the stream larger than $1/\sqrt{t_j+1}$ times the $\ell_2$-mass of all but the $t_j$ heaviest elements), for instance by \cref{lem:heavy-hitters}. They send the heavy entries to the coordinator, who reports the union of the all the received entries.
\end{enumerate}

\subparagraph*{Analysis of Communication, Space, and Time}
The algorithm uses 4 communication rounds. The total communication (in words) is $O(d \log^2 n)$ from sending $S$ and the vectors $S (X e_j)$, plus $O(\epsilon^{-2} \log^2 n)$ from sending the heavy hitters; their total number is indeed
\begin{equation*}
    O\parens*{\sum_j t_j \log^2 n} = O(k \log^2 n) = O(\epsilon^{-2} \log^2 n).
\end{equation*}
The total space is $O(d \log^2 n)$ from each player $j$ storing $S$, the vector $S (X e_j)$, and the matrix $M^{(\ell_j)}$, plus $O(\epsilon^{-2} \log^2 n)$ from storing the vectors $M^{(\ell_j)} (X e_j)$. 
The total time is $O(d \log^4 n + \epsilon^{-2} \log^2 n)$ from \cref{thm:budget-alloc}, plus $O(u \log^2 n)$ from the passes over the streams of total length $u$. Note that each update in the stream involves exactly one column of $S$ or the heavy hitter sketch, and we can compute the non-zero entries of each such column in time $O(\log^2 n)$, therefore each update can be performed in time $O(\log^2 n)$. The remaining steps are dominated by this running time.

\subparagraph*{Correctness}
Since player $j$ reports $O(t_j)$ heavy hitters, the number of heavy hitters ultimately reported by the coordinator is $\sum_j O(t_j) \leq O(k) = O(\epsilon^{-2})$. It remains to prove that with high probability we report all heavy hitters of $X$. Suppose that the $i$-th entry in the stream of player~$j$ satisfies $|X[i, j]| \geq \epsilon \norm{X_{-k}}_F$. In this case the entry also exceeds $\norm{X_{-k}}_F^2 / (k+1)$, and thus, by \cref{thm:budget-alloc}, also $\norm{(X e_j)_{-t_j}}_2^2 / (t_j + 1)$. In other words, the $i$-th item in the stream is a heavy hitter that is reported with high probability in step 4.

\end{document}